%% file: main_maxip.tex
\documentclass[11pt]{article}
\usepackage{amsfonts,amsmath,amsthm,amssymb,dsfont,etex,framed}

\usepackage{microtype}
\usepackage{tikz}
\usetikzlibrary{arrows,automata}
\usepackage{comment, times}
\usepackage[margin=1in]{geometry}
\usepackage{soul,color}
\usepackage{graphicx,float,wrapfig}
\usepackage{mathrsfs}
\usepackage[usenames,dvipsnames]{pstricks}
\usepackage{epsfig}
\usepackage{pst-grad} 
\usepackage{pst-plot} 
\usepackage[linesnumbered,boxed,ruled,vlined]{algorithm2e}
\usepackage{tablefootnote}

\usepackage{url}

\usepackage{tabu}

\usepackage[normalem]{ulem}

\newtheorem{theo}{Theorem}[section]

\newtheorem{lemma}[theo]{Lemma}

\newtheorem{prop}[theo]{Proposition}
\newtheorem{cor}[theo]{Corollary}
\newtheorem{defi}[theo]{Definition}
\newtheorem{rem}[theo]{Remark}

\newenvironment{proofof}[1]{\begin{proof}[Proof of #1]}{\end{proof}}

\newenvironment{reminder}[1]{\bigskip
	\noindent {\bf Reminder of #1  }\em}{\smallskip}

\newcommand{\R}{\mathbb{R}}
\newcommand{\Ex}{\mathbb{E}}

\newcommand{\distr}{\mathcal{D}}
\newcommand{\alg}{\mathbb{A}}

\newcommand{\poly}{\operatorname*{poly}}

\newcommand{\BQP}{\mathsf{BQP}}

\newcommand{\AM}{\mathsf{AM}}
 
\newcommand{\MA}{\mathsf{MA}}

\newcommand{\PTIME}{\mathsf{P}}

\newcommand{\NP}{\mathsf{NP}}
\newcommand{\coNP}{\mathsf{coNP}}

\newcommand{\UPP}{\mathsf{UPP}}
 
\newcommand{\cc}{\mathsf{cc}}
\newcommand{\eps}{\epsilon}

\newcommand{\polylog}{\operatorname*{polylog}}

\renewcommand{\epsilon}{\varepsilon}

\def\ShowAuthNotes{1}
\ifnum\ShowAuthNotes=1
\newcommand{\authnote}[2]{\ \\ \textcolor{red}{\parbox{0.9\linewidth}{[{\footnotesize {\bf #1:} { {#2}}}]}}\newline}
\else
\newcommand{\authnote}[2]{}
\fi

\newcommand{\WO}{\widetilde{O}}

\let\svfootnoterule\footnoterule
\renewcommand\footnoterule{\vfill\svfootnoterule}
\newcommand{\SAT}{\textsf{SAT}}
\newcommand{\SETH}{\textsf{SETH}}

\newcommand{\IntOV}{\textsf{$\mathbb{Z}$-OV}}
\newcommand{\Hopcroft}{\IntOV}

\newcommand{\OPT}{\textsf{OPT}}

\newcommand{\WT}{\widetilde}

\newcommand{\OV}{\textsf{OV}}
\newcommand{\ExactIP}{\textsf{Exact-IP}}
\newcommand{\MaxIP}{\textsf{Max-IP}}
\newcommand{\MinIP}{\textsf{Min-IP}}

\newcommand{\MAX}{\textsf{Max}}
\newcommand{\MIN}{\textsf{Min}}

\newcommand{\WMAX}{\widetilde{\MAX}}
\newcommand{\WMIN}{\widetilde{\MIN}}

\newcommand{\MAMaxIP}{\textsf{$\WMAX$-IP}}
\newcommand{\MAMinIP}{\textsf{$\WMIN$-IP}}

\renewcommand{\MAMaxIP}{\textsf{Apx-Max-IP}}
\renewcommand{\MAMinIP}{\textsf{Apx-Min-IP}}

\newcommand{\highlight}[1]{\medskip \noindent \textbf{#1}}

\newcommand{\MAXSAT}{\textsf{MAX-$\SAT$}}

\newcommand{\BCPp}[1][p]{\textsf{Bichrom.-$\ell_#1$-Closest-Pair}}
\newcommand{\FPp}[1][p]{\textsf{$\ell_#1$-Furthest-Pair}}

\newcommand{\BCP}{\textsf{Bichrom.-Closest-Pair}}
\newcommand{\FP}{\textsf{Furthest-Pair}}

\newcommand{\sat}{\textsf{sat}}
\newcommand{\val}{\textsf{val}}
\newcommand{\cla}{\mathcal{C}}

\newcommand{\SATPAIR}{\textsf{-Satisfying-Pair}}

\newcommand{\JaccardIndexPair}{\textsf{Jaccard-Index-Pair}}

\newcommand{\dist}{\textsf{dist}}

\newcommand{\thOV}{\textsf{3-OV}}
\newcommand{\IP}{\textsf{IP}}

\title{An Equivalence Class for Orthogonal Vectors\thanks{Supported by NSF CCF-1741615 (CAREER: Common Links in Algorithms and Complexity). Any opinions, findings and conclusions or recommendations expressed in this material are those of the authors and do not necessarily reflect the views of the National Science Foundation.}}

\author{Lijie Chen\\MIT
	\and Ryan Williams\\MIT}

\date{}

\begin{document}
	\maketitle
	
	\begin{abstract}
		The Orthogonal Vectors problem ($\OV$) asks: \emph{given $n$ vectors in $\{0,1\}^{O(\log n)}$, are two of them orthogonal?} $\OV$ is easily solved in $O(n^2 \log n)$ time, and it is a central problem in fine-grained complexity: dozens of conditional lower bounds are based on the popular hypothesis that $\OV$ cannot be solved in (say) $n^{1.99}$ time. However, unlike the APSP problem, few other problems are known to be non-trivially \emph{equivalent} to $\OV$.
		
		We show $\OV$ is truly-subquadratic equivalent to several fundamental problems, all of which (a priori) look harder than $\OV$. A partial list is given below:
		\begin{enumerate}
			
			\item ($\MinIP/\MaxIP$) Find a red-blue pair of vectors with minimum (respectively, maximum) inner product, among $n$ vectors in $\{0,1\}^{O(\log n)}$.
			
			\item ($\ExactIP$) Find a red-blue pair of vectors with inner product equal to a given target integer, among $n$ vectors in $\{0,1\}^{O(\log n)}$.
			
			\item ($\MAMinIP/\MAMaxIP$) Find a red-blue pair of vectors that is a 100-approximation to the minimum (resp. maximum) inner product, among $n$ vectors in $\{0,1\}^{O(\log n)}$.
			
			\item (Approximate $\BCPp$) Compute a $(1 + \Omega(1))$-approximation to the $\ell_p$-closest red-blue pair (for a constant $p \in [1,2]$), among $n$ points in $\R^d$, $d \le n^{o(1)}$.
			
			\item (Approximate $\FPp$) Compute a $(1 + \Omega(1))$-approximation to the $\ell_p$-furthest pair (for a constant $p \in [1,2]$), among $n$ points in $\R^d$, $d \le n^{o(1)}$.
		\end{enumerate}
		Therefore, quick constant-factor approximations to maximum inner product imply quick \emph{exact} solutions to maximum inner product, in the $O(\log n)$-dimensional setting. Another consequence is that the ability to find vectors with zero inner product suffices for finding vectors with maximum inner product. 
		
		Our equivalence results are robust enough that they continue to hold in the data structure setting. In particular, we show that there is a $\poly(n)$ space, $n^{1-\epsilon}$ query time data structure for \emph{Partial Match} with vectors from $\{0,1\}^{O(\log n)}$ if and only if such a data structure exists for $1+\Omega(1)$ \emph{Approximate Nearest Neighbor Search} in Euclidean space.
		
		To establish the equivalences, we introduce two general frameworks for reductions to $\OV$: one based on $\Sigma_2$ communication protocols, and another based on locality-sensitive hashing families. 
		
		In addition, we obtain an $n^{2 - 1/O(\log c)}$ time algorithm for $\MAMinIP$ with $n$ vectors from $\{0,1\}^{c\log n}$, matching state-of-the-art algorithms for $\OV$ and $\MAMaxIP$. As an application, we obtain a faster algorithm for approximating ``almost solvable'' $\MAXSAT$ instances.
	\end{abstract}

	\input{intro}
	\input{prelim}

	\input{Sigma2}
	\input{equiv-class}
	\input{MDOV}
	\input{tighter}
	\input{data-structure}
	\input{Apx-Min-IP}
	
	\section*{Acknowledgments}
	We thank Virginia Vassilevska Williams for many comments on an early draft of this paper. We are grateful to Josh Alman, Jiawei Gao, Ofer Grossman, Kaifeng Lyu, Karthik {C. S.}, Ruosong Wang, Virginia Vassilevska Williams, Grigory Yaroslavtsev and Peilin Zhong for helpful discussions during the work. We are especially grateful to Aviad Rubinstein for several helpful discussions which inspired the discovery of the equivalence between $\BCP$ and $\OV$.
	
	\bibliographystyle{alpha}
	
	\newcommand{\etalchar}[1]{$^{#1}$}

	\appendix		
	\input{app}
	\input{moreapp}
\end{document}

%% file: intro.tex
\section{Introduction}
	
	Fine-grained complexity asks: \emph{what is the ``correct'' exponent in the running time of a given problem?} For a problem known to be solvable in time $t(n)$, can it be solved in time $t(n)^{1-\eps}$, for a constant $\eps > 0$? If not, can we give evidence that such an improvement is impossible? In recent years, based on several conjectures such as the Orthogonal Vectors Conjecture (OVC) (implied by the Strong Exponential Time Hypothesis, a.k.a. SETH\footnote{The Strong Exponential Time Hypothesis (SETH) states that for every $\eps > 0$ there is a $k$ such that $k$-$\SAT$ cannot be solved in $O((2-\eps)^n)$ time~\cite{IP01-SETH}.}), the APSP Conjecture and the $k$-Sum Conjecture, tight conditional polynomial-time lower bounds have been established for problems in $\PTIME$ from many areas of computer science.
	
	In a nutshell, results in the \emph{Fine-Grained Complexity} program begin with the conjecture that it is \emph{hard to improve the runtime exponent} of some problem $\Pi_{\textsf{hard}}$, and show it is also hard to improve the exponent of another problem $\Pi$, by constructing a ``fine-grained'' reduction from $\Pi_{\textsf{hard}}$ to $\Pi$. This is similar to the situation with $\NP$-completeness, where one shows a problem $\Pi$ is ``hard'' by giving a polynomial-time reduction from another $\NP$-complete problem to $\Pi$.
	
	A crucial conceptual difference between the Fine-Grained Complexity program and $\NP$-hardness is that \emph{all of the thousands of known $\NP$-complete problems form an equivalence class}: there is either a polynomial-time algorithm for all of them, or no polynomial-time algorithm for any of them. In contrast, with Fine-Grained Complexity, few equivalence classes are known, especially for those numerous problems whose hardnesses are based on the SETH/OVC (a notable exception is the equivalence class for APSP~\cite{WW10-subcubic,williamssome}; see the related works section for more details). 
    
To give three (out of many examples), it is known that Edit Distance~\cite{BI15}, Frechet Distance~\cite{bringmann2014walking}, and computing the diameter of a sparse graph~\cite{RV13} cannot be done in $n^{2-\delta}$ time for any $\delta > 0$, assuming the following problem is not in $n^{2-\eps}$ time for a universal $\eps > 0$:
    
\begin{framed}
{\bf Orthogonal Vectors} ($\OV$): Given $n$ vectors in $\{0,1\}^d$ where $d = O(\log n)$, are there two vectors with inner product zero?
\end{framed}

However, it is not known if Edit Distance, Frechet Distance, or Diameter are \emph{equivalent} to $\OV$, in any interesting sense. 

Prior work has established an equivalence class for ``moderate-dimensional $\OV$'', where the vector dimension $d = n^{\delta}$ for a constant $\delta > 0$~\cite{gao2017completeness}. In particular, this version of $\OV$ is equivalent to various sparse graph and hypergraph problems. It seems likely that ``moderate-dimensional $\OV$'' is much more difficult to solve than the ``low-dimensional'' setting of $d = O(\log n)$ as defined above, and the SETH already implies that the low-dimensional case is difficult~\cite{Wil05,williams2014finding}. Thus the problem of establishing an equivalence class for ``low-dimensional'' $\OV$ is an interesting one.


\subsection{An Equivalence Class for Sparse Orthogonal Vectors}	
	
	Our first result is an interesting equivalence class for Orthogonal Vectors in the $O(\log n)$-dimensional setting. To formally state our results, we begin with some notation. 
	
	\begin{itemize}
		\item 
		For a problem $\Pi$ on Boolean vectors, we say \emph{$\Pi$ is in truly subquadratic time} if there is an $\eps > 0$ such that for all constant $c$, $\Pi$ is solvable in $O(n^{2-\eps})$ time on $n$ vectors in $c \log n$ dimensions. Note the Orthogonal Vectors Conjecture (OVC) is equivalent to saying ``$\OV$ is not in truly subquadratic time.''
		\item 
		For a problem $\Pi$ on real-valued points, we say  \emph{$\Pi$ can be approximated in truly subquadratic time}, if there is a $\delta > 0$ such that for all $\eps > 0$, a $(1+\eps)$ approximation to $\Pi$ is computable in $O(n^{2-\delta})$ time.
		\item 
		For a problem $\Pi$ with output in $[0,L]$ (for a parameter $L$), we say \emph{$\Pi$ can be additively approximated in truly subquadratic time}, if there is a $\delta > 0$ such that for all $\eps > 0$, an $\epsilon \cdot L$ additive approximation to $\Pi$ is computable in $O(n^{2-\delta})$ time.
	\end{itemize}
	
	\begin{theo}\label{theo:equiv-class-OV}
		The following problems are either all in (or can be approximated in) truly subquadratic time, or none of them are:\footnote{A list of formal definitions of the these problems can be found in Definition~\ref{defi:prob-list}.}	
			
		\begin{enumerate}
			\item ($\OV$) Finding an orthogonal pair among $n$ vectors.
			
			\item ($\MinIP/\MaxIP$) Finding a red-blue pair of vectors with minimum (respectively, maximum) inner product, among $n$ vectors.
			
			\item ($\ExactIP$) Finding a red-blue pair of vectors with inner product exactly equal to a given integer, among $n$ vectors.
			
			\item ($\MAMinIP/\MAMaxIP$) Finding a red-blue pair of vectors that is a 100-approximation to the minimum (resp. maximum) inner product, among $n$ vectors.\footnote{The constant $100$ can be replaced by any fixed constant $\kappa > 1$.}
			
			\item (Approximate Bichrom. $\ell_p$-Closest Pair) Approximating the $\ell_p$-closest red-blue pair (for a constant $p \in [1,2]$), among $n$ points.
			
			\item (Approximate $\ell_p$-Furthest Pair) Approximating the $\ell_p$-furthest pair (for a constant $p \in [1,2]$), among $n$ points.
			
			\item (Approximate Additive $\MaxIP$) Additively approximating the maximum inner product of all red-blue pairs, among $n$ vectors.
			
			\item (Approximate $\JaccardIndexPair$) Additively approximating the maximum Jaccard index\footnote{see Theorem~\ref{defi:prob-list-geo} for a formal definition} between $a \in A$ and $b \in B$, where $A$ and $B$ are two collections of $n$ sets.
			
		\end{enumerate}
	\end{theo}

	For approximate additive $\MaxIP$, $L$ (the additive approximation parameter) is simply the dimensions of the vectors, while for approximate $\JaccardIndexPair$, $L$ is $1$. For $\Pi$ among the first four problems listed above, we use the notation $\Pi_{n,d}$ to denote $\Pi$ with $n$ vectors from $\{0,1\}^{d}$. \footnote{In the paper we will consider red-blue version for all the above problems, and $\Pi_{n,d}$ denotes $\Pi$ with two sets of $n$ vectors from $\{0,1\}^d$.} For the last four problems, we assume the dimensions (or the size of the sets) and the bit complexity of the points are $n^{o(1)}$ throughout the paper.
	
	Prior work showed $\OV$ is equivalent to Dominating Pair\footnote{Given two sets $A,B$ of vectors from $\mathbb{R}^{O(\log n)}$, find $(a,b) \in A \times B$ such that $b$ \emph{dominates} $a$ (that is, $b_i  > a_i$ for all $i$).}~\cite{Chan17a} and other simple set problems~\cite{BorassiCH16}; our results add several interesting new members into the equivalence class. All problems listed above were already known to be $\OV$-hard~\cite{Wil05,AW15,Rub18BetterMA}. Our main contribution here is to show that \textbf{\emph{they can all be reduced back to $\OV$}}. For example, detecting an orthogonal \emph{Boolean} pair ($\OV$) is equivalent to approximating the distance between two sets of points in $\R^{n^{o(1)}}$ ($\BCP$)! 
	
	In previous works~\cite{gao2017completeness,ABDN2018MoreCons}, several general techniques are given for constructing reductions to $\OV$. These papers focus on the ``moderate-dimensional'' setting, and their reductions can not be used directly in the ``sparse'' $O(\log n)$ dimensional setting here.	
	
	\paragraph*{Our Techniques: Two Reduction Frameworks for $\OV$.} In order to construct reductions to $O(\log n)$ dimensional $\OV$, we propose the following two general frameworks.
	
	\begin{itemize}
		\item \textbf{$\Sigma_2^\cc$ Protocols.} Inspired by previous works on the connections between communication complexity and fine-grained complexity~\cite{ARW17-proceedings,Rub18BetterMA,karthik2017parameterized,abboud2018fast,chen2018hardness,CGLRR18Meets}, we draw another connection along this line, showing that an efficient $\Sigma_2^\cc$ protocol\footnote{see Definition~\ref{defi:Sigma-2-communication-protocol} for a formal definition} for a function $F$ implies a reduction from a related problem to $\OV$. We use this technique to establish the equivalences among the first four problems in Theorem~\ref{theo:equiv-class-OV}.
		
		\item \textbf{Locality-sensitive Hashing Families (LSH).} To show equivalences between $\OV$ and the last four approximation problems, we apply known tools from \emph{locality-sensitive hashing}. In particular, we show that for any metric admitting an efficient LSH family, finding the closest bichromatic pair or the furthest pair w.r.t. this metric can be reduced to $\MAMaxIP$, which can in turn be reduced to $\OV$.
	\end{itemize}
	
	We remark that there are no non-trivial lower bounds known against $\Sigma_2^\cc$ protocols~\cite{GPW18CCLandscape}, which suggests that $\Sigma_2^\cc$ protocols \textbf{\em could be very powerful}, and the first approach (Theorem~\ref{theo:Sigma2-to-OV}) may be applicable to many other problems. This is not the case for $\MA^\cc$ protocols which were used in several previous works~\cite{ARW17-proceedings,Rub18BetterMA,karthik2017parameterized,chen2018hardness}: for example, there is an essentially tight $\Omega(\sqrt{n})$ $\MA^\cc$ lower bound for Set-Disjointness~\cite{Klauck2003MALowB,AW09-algebrization,chen2018hardness}. These two frameworks are discussed in Section~\ref{sec:intro-framework} in detail.	
	
	\paragraph*{Equivalence Between Partial Match and Approximate Nearest Neighbor Search.} Our reductions are robust enough that they also hold in the data structure setting. In particular, consider the following two fundamental data structure problems:
	
	\begin{itemize}
		\item \textbf{Partial Match:} Preprocess a database $\mathcal{D}$ of $n$ points in $\{0,1\}^d$ such that, for all query of the form $q \in \{0,1,\star\}^{d}$, either report a point $x \in D$ matching all non-$\star$ characters in $q$ or report that no $x$ exists.
		
		
		\item \textbf{Approximate Nearest Neighbor Search (NNS) in $\ell_p$ space:} Preprocess a database $\mathcal{D}$ of $n$ points from $\mathbb{R}^{m}$ such that, for all query point $x \in \mathbb{R}^m$, one can find a point $y \in \mathcal{D}$ such that $\|x-y\|_{p} \le (1+\epsilon) \cdot \min_{z \in \mathcal{D}} \| x - z \|_p$.
		
	\end{itemize}
	\begin{rem}
	We remark that Partial Match is known to be equivalent to an online version of $\OV$~\cite{abboud2015more} (see also Section~\ref{sec:data-structure}), and NNS in $\ell_p$ space is simply the online version of $\BCPp[p]$.
	\end{rem}
	
	Partial Match has been studied extensively for decades (see e.g. Rivest's PhD thesis~\cite{rivest1974analysis}). However, the algorithmic progress beyond trivial solutions (building a look-up table of size $2^{\Omega(d)}$, or trying all $n$ points on each single query) have been quite limited. It is generally believed that it is intractable when $d$ is large enough. Many unconditional lower bounds are known in the cell-probe model~\cite{MiltersenNSW98,BorodinOR99,JayramKKR04,PanigrahyTW08,PatrascuT09}, but the gap between the best data structures~\cite{CharikarIP02,ColeGL04} and known lower bounds remains very large. 
	
	Approximate Nearest Neighbor Search has a wide range of applications in computing, including machine learning, computer vision, databases and others (see~\cite{AndoniI08,Moraleda08} for an overview).
Tremendous research effort has been devoted to this problem (see e.g. the recent survey of~\cite{andoni2018approximate} and Razenshteyn's
	PhD thesis~\cite{Razenshteyn17}). Yet all known algorithms exhibit a query time of at least $n^{1 - O(\epsilon)}$ when the approximation ratio is $1+\epsilon$, approaching the brute-force query time $n$ when $\epsilon$ goes to $0$. 
	
	In general, whether there is a \emph{polynomial} space, $n^{1-\delta}$ query time data structure for Partial Match for all $d = O(\log n)$, or Approximate NNS for all constant approximation ratio $> 1$ are two long-standing open questions.\footnote{Under $\SETH$, it is shown that there is no such data structure with \emph{polynomial pre-processing time}~\cite{AhlePR016,Wil05,Rub18BetterMA}.} We show these two questions are \textbf{\emph{equivalent}}.

	\begin{theo}\label{theo:eq-partial-match-NNS}
		The following are equivalent:		
		\begin{itemize}
			\item There is a $\delta > 0$ such that for all constant $c$, there is a data structure for Partial Match with string length $d = c \log n$ that uses $\poly(n)$ space and allows $n^{1-\delta}$ query time.
			
			\item There is a $\delta > 0$ such that for all $\epsilon > 0$, there is an data structure for Approximate NNS in $\ell_p$ with approximation ratio $(1+\epsilon)$ that uses $\poly(n)$ space and allows $n^{1-\delta}$ query time, for some constant $p \in [1,2]$.
		\end{itemize}		
	\end{theo}

	
	\paragraph*{Tighter Connection Between $\MaxIP$, Bichrom. $\ell_p$-Closest Pair and $\ell_p$-Furthest Pair.} For a subset of problems in Theorem~\ref{theo:equiv-class-OV}, we can establish even tighter reductions.
	
	The state-of-the-art algorithm for $(1+\eps)$ approximation to $\BCPp$ runs in $n^{2 - \WO(\eps^{1/3})}$ time, and for $\MaxIP_{n,c\log n}$, the best running time $n^{2 - \WO(1/\sqrt{c})}$. Both algorithms are presented in~\cite{alman2016polynomial}, and relied on probabilistic threshold functions. 
	
	Comparing to the $n^{2-1/O(\log c)}$ time algorithm for $\OV_{n,c\log n}$~\cite{abboud2015more,chan2016deterministic}, the dependence on $c$ or $\eps$ in these two algorithms are much worse, rendering them useless when $\eps^{-1}$ or $c$ are $\log^{\omega(1)} n$. So it is natural to ask whether the dependence can be improved to at least sub-polynomial in $\eps$ and $c$, i.e. $n^{2 - 1/c^{o(1)}}$ or $n^{2 - \eps^{o(1)}}$.
	
	We show that a modest improvement on the running time dependence on $\eps$ or $c$ for any of the following problems directly implies similar improvements for other problems as well.
	
	\begin{theo}\label{theo:tighter-reduction}
		The following are equivalent:
		\begin{itemize}
			\item An $\eps \cdot d$ additive approximation to $\MaxIP_{n,d}$ is computable in $n^{2 - \eps^{o(1)}}$ time.
			\item $\MaxIP_{n,c\log n}$ is solvable in $n^{2 - 1/c^{o(1)}}$ time.
			\item $\ExactIP_{n,c\log n}$ is solvable in $n^{2 - 1/c^{o(1)}}$ time.
			\item A $(1+\eps)$ approximation to $\BCPp[p]$ is computable in $n^{2 - \eps^{o(1)}}$ time (for a constant $p \in [1,2]$).
			\item A $(1+\eps)$ approximation to $\FPp$ is computable in $n^{2 - \eps^{o(1)}}$ time (for a constant $p \in [1,2]$).
		\end{itemize}
	\end{theo}

	In~\cite{Rub18BetterMA} (Theorem~4.1), it is implicitly shown that $\ExactIP_{n, c\log n}$ can be reduced to $(1 + 1/\exp(c))$ approximating $\BCPp$. This suffices for the case when $c$ is a constant (which is needed for Theorem~\ref{theo:equiv-class-OV}), but falls short of proving the above tighter connections. 
	
	In a nutshell, \cite{Rub18BetterMA}'s reduction applies a very efficient $\MA$ protocol for Set-Disjointness using AG-codes, and it uses ``brute-force'' gadgets to simulate an inner product between two short vectors in $\mathbb{F}_{q^2}$.	We improve~\cite{Rub18BetterMA}'s reduction by carefully modifying its $\MA$ protocol, and replacing its brute-force gadgets by a more efficient one. Informally, our theorem shows $\ExactIP_{n, c\log n}$ can be reduced to $(1 + 1/\poly(c))$ approximating $\BCP$ (see Lemma~\ref{lm:ExactIP-to-additive-MaxIP} and Lemma~\ref{lm:prev-work-tt}), which is an exponential improvement over the old reduction.

	
		
	\paragraph*{Equivalence Results in the Moderate Dimensional Setting.} Theorem~\ref{theo:equiv-class-OV} establishes an equivalence class for the sparse $O(\log n)$ dimensional setting. It is natural to ask whether the equivalence continues to hold in the moderate dimensional case as well. 
	
	Unfortunately, an unusual (and interesting) property of our reduction used in Theorem~\ref{theo:equiv-class-OV} is that it blows up $c$ (the constant before $\log n$) exponentially, and creates multiple instances. That is, an $\ExactIP$ instance with $c\log n$ dimensions is reduced to \emph{many} $\OV$ instances with $\exp(c) \log n$ dimensions (see the proof of Lemma~\ref{lm:ExactIP-to-OV}). This renders the reduction useless in the moderate-dimensional setting, where $c$ could be as large as $n^\delta$.
	
	Still, using different techniques, we obtain some additional equivalence results in the moderate dimensional setting. For a problem $\Pi$ on Boolean vectors, we say that \emph{moderate dimensional $\Pi$ is in truly subquadratic time}, if there are two constants $\eps,\delta > 0$ such that $\Pi$ is solvable in $n^{2-\eps}$ time on $n$ vectors with $n^\delta$ dimensions.
	
	\begin{theo}\label{theo:equiv-OV-MAMinIP-MD}
		Moderate dimensional $\OV$ is in truly subquadratic time if and only if moderate dimensional $\MAMinIP$ is.
	\end{theo}

	\begin{theo}\label{theo:equiv-MaxIP-MinIP-ExactIP-MD}
		For moderate dimensional $\MaxIP$, $\MinIP$, and $\ExactIP$, either all of them are in truly subquadratic time, or none of them are.
	\end{theo}

	To show moderate dimensional $\OV$ and $\MAMinIP$ are equivalent, we use a sophisticated reduction which is partially inspired by the classical Goldwasser-Sipser $\AM$ protocol for approximate counting~\cite{GoldwasserS89} (see the proof of Lemma~\ref{lm:reduction-MAMinIP-to-OV} for details). For $\MaxIP$, $\MinIP$ and $\ExactIP$, we apply some folklore encoding tricks.

	It is an interesting open question that whether these two separate equivalence classes can be merged into one. In particular, \emph{is moderate dimensional $\OV$ equivalent to moderate dimensional $\MaxIP$?}
	
	An immediate corollary of Theorem~\ref{theo:equiv-OV-MAMinIP-MD} is that it adds $\MAMinIP$ as a new member to the equivalence class of moderate dimensional $\OV$ established in~\cite{gao2017completeness}.
	

	\subsection{New Algorithms for $\MAMinIP$ and $\MAMaxIP$}
	
	It was recently shown in~\cite{chen2018hardness} that $\MAMaxIP$ can be solved in $n^{2 - 1 / O(\log c)}$ time, while the best known algorithm for solving $\MAMinIP$ just applies the $n^{2 - 1/\widetilde{O}(\sqrt{c})}$ time algorithm for $\MinIP$~\cite{alman2016polynomial}. We show that in fact we can derive an algorithm with similar running time for $\MAMinIP$ as well. 
	
	\begin{theo}\label{theo:fast-algo-MAMinIP}
		There are $n^{2 - 1/O(\log c)}$ time randomized algorithms for $\MAMinIP_{n,c \log n}$ and $\MAMaxIP_{n,c \log n}$.
	\end{theo}
	
	\begin{rem}
		Our new algorithm works equally well for $\MAMaxIP$. Hence, we provide a different $n^{2 - 1 / O(\log c)}$ time algorithm for $\MAMaxIP$ than~\cite{chen2018hardness}. One caveat here is that our algorithms are \emph{randomized}, while the algorithms in~\cite{chen2018hardness} are \emph{deterministic}.
	\end{rem}
	
	The algorithms are based on the polynomial method: we construct a low-degree probabilistic polynomial over $\mathbb{F}_2$ for functions closely related to $\MAMinIP$ and $\MAMaxIP$, and the rest follows from the framework of~\cite{abboud2015more}. 
	
	\highlight{Application: A Fast Algorithm for Approximating ``Almost Solvable'' $\MAXSAT$ Instances.} Here we give an application of our $\MAMinIP$ algorithm. For a $\MAXSAT$ instance $\varphi$ with $m$ clauses, we denote $\OPT(\varphi)$ to be the maximum number of the clauses that can be satisfied, and $\sat(\varphi) := \OPT(\varphi) / m$. 
	
	\begin{theo}\label{theo:fast-algo-MAXSAT-eps}
		Let $\varphi$ be a $\MAXSAT$ instance on $n$ variables with $m$ clauses, and $\eps = 1 - \sat(\varphi)$. There is a $2^{n (1 - 1 / O(\log \eps^{-1}) )}$ time algorithm to find an assignment $x$ satisfying at least $(1 - 2 \eps) \cdot m$ clauses\footnote{$(1 - 2\eps)$ can be replaced by $(1 - \kappa \eps)$ for any constant $\kappa > 1$.}.
	\end{theo}
	
	That is, when $\varphi$ is ``almost solvable'' ($\sat(\varphi)$ is very close to $1$), we have a fast algorithm to compute an approximate solution $x$, which is only a ``little'' worse than the optimal solution.
	The following corollary is immediate.
	
	\begin{cor}
		Let $\varphi$ be a $\MAXSAT$ instance on $n$ variables and $\eps \in (0,1/10)$. Given the promise that either $\sat(\varphi) \ge 1 - \eps$ or $\sat(\varphi) < 1 - 2\eps$, there is a $2^{n (1 - 1/O(\log \eps^{-1}))}$ time algorithm for deciding which is the case.
	\end{cor}

	The best known previous algorithm for the above problem requires at least $2^{n (1- \eps^{1/3})}$ time~\cite{alman2016polynomial}, in which case the dependence on $\eps$ is exponentially worse than our new algorithm. In particular, it fails to give any improvement when $\eps < 1 / n^{3}$, while our algorithm is faster than brute-force even if $\eps = 1/2^{n^{0.99}}$.

	\subsection{Techniques: Two General Frameworks for Establishing $\OV$ Equivalence}\label{sec:intro-framework}
	
	
	In the following we discuss two general frameworks for reductions to $\OV$.	To state our results formally, we first define the $F\SATPAIR$ problem for a problem $F$.\footnote{This notation is borrowed from~\cite{AbboudHWW16}, which studied the Satisfying Pair problem for Branching Programs.}
	
	\begin{defi}[\cite{AbboudHWW16}]\label{defi:F-SatPair}
		Let $F: \{0,1\}^{d} \times \{0,1\}^{d} \to \{0,1\}$, $F\SATPAIR_{n}$ is the problem: given two sets $A$ and $B$ of $n$ vectors from $\{0,1\}^d$, determine whether there is a pair $(a,b) \in A \times B$ such that $F(a,b) = 1$.
	\end{defi}
	\begin{rem}
		For example, let $F_{\OV}$ be the function checking whether two vectors from $\{0,1\}^d$ are orthogonal. Then, $F_{\OV}\SATPAIR_n$ is simply $\OV_{n,d}$.
	\end{rem}
	 
	\subsubsection{$\Sigma_2$ Communication Protocols and Reductions to Orthogonal Vectors}
	
	Our first framework is based on $\Sigma_2$ communication protocols ($\Sigma_2^\cc$ protocols). We begin with a formal definition of such protocols.
	
	\newcommand{\rand}{\textsf{rand}}
	
	\begin{defi}[$\Sigma_2^\cc$ Protocol~\cite{BFS86CC}]\label{defi:Sigma-2-communication-protocol}
		Let $F : \mathcal{X} \times \mathcal{Y} \to \{0,1\}$ be a function. A $\Sigma_2^\cc$ protocol $\Pi$ for $F$ is specified as follows:
		\begin{itemize}
			\item There are two players, Alice holds input $x \in \mathcal{X}$ and Bob holds input $y \in \mathcal{Y}$.
			
			\item There are two provers Merlin and Megan. 
						
			\item Merlin sends a string $a \in \{0,1\}^{m_1}$ and Megan sends a string $b \in \{0,1\}^{m_2}$ (which are functions of both $x$ and $y$) to both Alice and Bob. Then Alice and Bob communicate $\ell$ bits with each other, and Alice decides whether to accept or reject the pair $(a,b)$. 
			\item $F(x,y) = 1$ if and only if there exists a string $a$ from Merlin, such that for all strings $b$ from Megan, Alice accepts $(a,b)$ after communications with Bob.
		\end{itemize}
		
		We say the protocol $\Pi$ is computationally-efficient, if both Alice and Bob's response functions can be computed in polynomial time with respect to their input length.
	\end{defi}

	We show that for any function $F$, if $F$ admits a certain efficient $\Sigma_2^\cc$ protocol, then $F\SATPAIR$ can be efficiently reduced to $\OV$. Formally, we have:

	\begin{theo}\label{theo:Sigma2-to-OV}
		Let $F : \{0,1\}^{d} \times \{0,1\}^{d} \to \{0,1\}$ and $n \in \mathbb{N}$, suppose $F$ has a computationally-efficient $\Sigma_2^\cc$ protocol, in which Merlin sends $m_1$ bits, Megan sends $m_2$ bits, and Alice and Bob communicate $\ell$ bits. 		Then there is a reduction from every $F\SATPAIR_{n}$ instance $I$ to $\OV_{n,2^{(m_2+\ell)}}$ instances $J_1,J_2,\dotsc,J_{2^{m_1}}$, such that $I$ is a yes instance if and only if there is a $j$ such that $J_j$ is a yes instance. The reduction takes $n \cdot 2^{O(m_1 + m_2 + \ell)} \cdot \poly(d)$ time.
	\end{theo}
	
	
	\paragraph*{Applications.} We use Theorem~\ref{theo:Sigma2-to-OV} to establish the equivalence between $\OV$, $\MinIP$ / $\MaxIP$, $\MAMaxIP$ / $\MAMinIP$ and $\ExactIP$. Previous works have established that $\OV$ can be reduced to all these problems, and that these problems can be reduced to $\ExactIP$. So it suffices for us to construct a reduction from $\ExactIP$ to $\OV$. Let the $\IP_{d,m} : \{0,1\}^{d} \times \{0,1\}^{d} \to \{0,1\}$ be the function that checks whether $\langle x, y \rangle = m$,  $\ExactIP$ is $\IP_{d,m}\SATPAIR$, so we can apply Theorem~\ref{theo:Sigma2-to-OV} with an efficient $\Sigma_2^\cc$ protocol for $\IP_{d,m}$. More applications can be found in the full version of the paper.
	
	\subsubsection{Locality-sensitive Hashing (LSH) Families and Reductions to Additive Approximation to $\MaxIP$}
	\newcommand{\Anew}{A_{\textsf{new}}}
	\newcommand{\Bnew}{B_{\textsf{new}}}
	
	To establish equivalence between $\OV$ and other approximation problems, we make use of a connection with LSH families. We begin with a generalized definition of an LSH family for a partial function. In the following, let $\mathcal{X}$ be an arbitrary set.
	
	\begin{defi}\label{defi:LSH-general}
		Let $f : \mathcal{X} \times \mathcal{X} \to \{0,1,\bot\}$\footnote{$f(x,y) = \bot$ means $f$ is ``undefined'' on $(x,y)$.}.
		We say $f$ admits a $(p_1,p_2)$-sensitive LSH family, if there is a family $\mathcal{F}$ of functions $h : \mathcal{X} \to \mathcal{S}$, such that for any $x,y \in \mathcal{X}$, a uniformly random function $h \in \mathcal{F}$ satisfies:
		\begin{itemize}
			\item If $f(x,y) = 1$, then $h(x) = h(y)$ with probability at least $p_1$.
			\item If $f(x,y) = 0$, then $h(x) = h(y)$ with probability at most $p_2$.
		\end{itemize}
		
		In addition, we require that $h$ can be efficiently drawn from $\mathcal{F}$, and $h(p)$ can be efficiently computed.\footnote{Being efficient here means the running time is polynomial in the bit complexity of the input.}
	\end{defi}
	
	The usual LSH families for a metric space are special cases of the above generalized definition.
	
	\begin{defi}\label{defi:LSH-metric}
		For a function $\dist : \mathcal{X} \times \mathcal{X} \to \mathbb{R}_{\ge 0}$, we say $\dist$ admits an LSH family, if for all $\eps > 0$ and real $R > 0$, there are two reals $p_1 = p_1(\eps)$ and $p_2 = p_2(\eps)$ such that the function $f^\dist_{R,(1+\eps)R} : \mathcal{X} \times \mathcal{X} \to \{0,1,\bot\}$ defined as
		\[
			f^\dist_{R,(1+\eps)R}(x,y) = \begin{cases}
			1 &\qquad \text{$\dist(x,y) \le R$,}\\
			0 &\qquad \text{$\dist(x,y) \ge (1+\eps) \cdot R$,}\\
			\bot &\qquad \text{otherwise,}
			\end{cases}
		\]
		admits a $(p_1,p_2)$-sensitive LSH family and $p_1 > p_2$.
	\end{defi}
	
	In particular, we show that an LSH family for a function implies a reduction to additively approximating $\MaxIP$, which can in turn be reduced to $\OV$. To formally state our reduction, we need to define $\mathcal{F}\SATPAIR$ for a partial function $\mathcal{F}$.
	
	\begin{defi}\label{defi:F-Sat-Pair-partial}		
		For a partial function $\mathcal{F} : \mathcal{X} \times \mathcal{X} \to \{0,1,\bot\}$, $\mathcal{F}\SATPAIR_n$ is the problem: given two sets $A,B \subseteq \mathcal{X}$ of size $n$, distinguish between the two cases:
		\begin{itemize}
			\item There is an $(x,y) \in A \times B$ such that $\mathcal{F}(x,y) = 1$.
			\item For all $(x,y) \in A \times B$, $\mathcal{F}(x,y) = 0$.
		\end{itemize}
	\end{defi}
	\begin{rem}
		Let $\mathcal{X}$ be $\R^{d}$, and set $\mathcal{F}(x,y) = 1$ for $\|x-y\| \le R$, $\mathcal{F}(x,y) = 0$ for $\|x-y\| \ge (1+\eps) \cdot R$ and undefined otherwise. Then $\mathcal{F}\SATPAIR$ distinguishes between the cases that the minimum distance between $A$ and $B$ is $\le R$ and $\ge (1+\eps) \cdot R$, which is the decision version of $(1+\eps)$-approximation to $\BCP$.
	\end{rem}

	Now we are ready to state our general reduction.

	\begin{theo}\label{theo:LSH-to-MaxIP-general}
		Suppose $f : \mathcal{X} \times \mathcal{X} \to \{0,1,\bot\}$ admits a $(p_1,p_2)$-sensitive LSH family. Let $\eps = p_1 - p_2$.
		
		Then there is a randomized reduction from $f\SATPAIR_{n}$ to computing an $\eps/8 \cdot d$ additive approximation to $\MaxIP_{n,d}$ with $d = O(\eps^{-2} \log n)$, which succeeds with probability at least $1 - 1/n$.		
	\end{theo}

	From Theorem~\ref{theo:LSH-to-MaxIP-general}, reductions from \textsf{Bichrom.-$\ell_2$-Closest-Pair} and $\FP$ to $\OV$ follows:

	\begin{cor}\label{cor:LSH-to-MaxIP-metric}		
		For a distance function $\dist : \mathcal{X} \times \mathcal{X} \to \mathbb{R}_{\ge 0}$ which admits an LSH family, $\BCP_{n,\dist}$ and $\FP_{n,\dist}$ can be approximated in truly subquadratic time if $\OV$ is in truly subquadratic time.
	\end{cor}

	\paragraph*{Applications.} We use Theorem~\ref{theo:LSH-to-MaxIP-general} and Corollary~\ref{cor:LSH-to-MaxIP-metric} to establish the equivalence between $\OV$ and all approximation problems listed in Theorem~\ref{theo:equiv-class-OV}. In particular, the $\ell_p$ metric and Jaccard Index admit efficient LSH families via $p$-stable distributions and the minHash method, which implies that they can be reduced to $\OV$ by Theorem~\ref{theo:LSH-to-MaxIP-general}.
	
	\subsection{Related Works}
	
	\paragraph{Equivalence Classes in Fine-Grained Complexity.} It is known that the All-Pairs Shortest Paths problem is sub-cubic time equivalent to many other problems~\cite{WW10-subcubic,backurs2016tight,abboud2015subcubic,lincoln2018tight}. A partial list includes: Negative Triangle, Triangle Listing, Shortest Cycle, 2nd Shortest Path, Max Subarray, Graph Median, Graph Radius and Wiener Index (see~\cite{williamssome} for more details on the APSP equivalence class). 
	
	In~\cite{gao2017completeness}, it is shown that ``moderate-dimensional'' $\OV$ (i.e., $\OV$ with $n^\delta$ dimensions for some $\delta > 0$) is equivalent to High-dimension Sparse $\OV$, High-dimension $2$-Set Cover, and High-dimension Sperner Family. It is also shown that for every $(k+1)$-quantifier first-order property, its model-checking problem can be reduced to Sparse $k$-$\OV$. In~\cite{CGLRR18Meets}, an equivalence class for $\textsf{Closest-LCS-Pair}$\footnote{ \textsf{Closest-LCS-Pair} is: given two sets $A,B$ of strings, compute $\max_{(a,b) \in A \times B} \textsf{LCS}(a,b)$.} is established, in particular, it shows $\textsf{Closest-LCS-Pair}$ and its (constant factor) approximate version are equivalent. In~\cite{CyganMWW17}, the authors present an equivalence class for $(\min,+)$-convolution, including some variants of the classical knapsack problem and problems related to subadditive sequences.
	
	\paragraph{Faster-Than-Brute-Force Algorithms for Problems in the Equivalence Class.} Most of the problems listed in Theorem~\ref{theo:equiv-class-OV} have  algorithms with some non-trivial speed-up depending on $c$ (when the dimension is $c \log n$) or $\eps$ (when the approximation ratio is $1+\eps$). Table~\ref{tab:speed-up} gives the state-of-the-art runtime bounds for these problems.
		
		\begin{table}[H]
			\begin{center}
				\begin{tabular}{|c|c|}
					\hline 
					Problem & $n^{2 - \delta}$ time, $\delta = f(c)$ or $f(\eps)$ \\ 
					\hline 
					$\OV$ & $1/O(\log c)$~\cite{abboud2015more,chan2016deterministic} \\ 
					$\MinIP$ \& $\MaxIP$ & $1/\widetilde{O}(\sqrt{c})$~\cite{alman2016polynomial}\\ 
					$\ExactIP$ &  $1/\widetilde{O}(c)$~\cite{AW15} \\ 
					$\MAMaxIP$ & $1/O(\log c)$~\cite{chen2018hardness} [This paper]  \\ 
					$\MAMinIP$ & $1/O(\log c)$ [This paper] \\ 
					$\textsf{B.-$\ell_2$-Closest-Pair}$ & $\widetilde{O}(\eps^{1/3})$~\cite{alman2016polynomial} \\
					$\FPp[p]$ & $\widetilde{O}(\eps^{1/3})$~\cite{alman2016polynomial}\tablefootnote{\cite{alman2016polynomial} only discussed $\BCPp[p]$ when $p \in \{1,2\}$, but one can observe that their algorithm in fact works equally well with $\BCPp[p]$ and $\FPp[p]$ for $p \in [1,2]$.} \\
					\hline 
				\end{tabular} 
			\caption{The best known running-time exponents for the problems shown (in this paper) to be equivalent to $\OV$.}\label{tab:speed-up}
			\end{center}
		\end{table}
	
	\paragraph*{Fine-Grained Complexity and Communication Complexity.}
	The connection between communication complexity and Fine-Grained Complexity dates back at least to~\cite{PW10}, in which it is shown that a sub-linear, computational efficient protocol for $3$-party Number-On-Forehead Set-Disjointness problem would refute SETH. The work of~\cite{ARW17-proceedings} shows hardness for approximate version for a host of important problems in $\PTIME$, using the $\WO(\sqrt{n})$ $\MA$ communication protocol for Set-Disjointness~\cite{AW09-algebrization}.
	
	Using Algebraic Geometry codes,~\cite{Rub18BetterMA} obtains a better $\MA$ protocol, which in turn improves the efficiency of the previous ``distributed PCP'' construction of~\cite{ARW17-proceedings}. He then shows $n^{2-o(1)}$-time hardness for $1+o(1)$-approximations to Bichromatic Closest Pair and $o(d)$-additive approximations to $\MaxIP_{n,d}$ with this new technique. \cite{karthik2017parameterized} use the Distributed PCP framework to derive inapproximability results for $k$-Dominating Set under various assumptions. In particular, building on the techniques of~\cite{Rub18BetterMA}, it is shown that under SETH, $k$-Dominating Set has no $(\log n)^{1/\poly(k,e(\eps))}$ approximation in $n^{k-\eps}$ time\footnote{where $e$ is a certain function from $ \mathbb{R}^{+} \to \mathbb{N}$}. 	
	
	\cite{abboud2018fast} make use of the $\widetilde{O}(\log n)$ $\textsf{IP}$ communication protocol for Set-Disjointness in~\cite{AW09-algebrization}, and shows a fast deterministic approximation algorithm to Longest Common Subsequence has interesting circuit lower bound consequences. Making use of the $\textsf{IP}$ communication protocol for low-space computation, \cite{CGLRR18Meets} establish an equivalence class for $\textsf{Closest-LCS-Pair}$.
	
	\cite{chen2018hardness} establishes a connection between hardness of the furthest pair problem in low dimensional Euclidean space and $\NP \cdot \UPP$ communication protocols for Set-Disjointness. He also shows the $\BQP$ communication protocol for Set-Disjointness~\cite{BCW98-quantum_communication} can be used to derive an inapproximability result for $\{-1,1\}\text{-}\MaxIP$.\footnote{the variant of $\MaxIP$ with vectors from $\{-1,1\}^d$ instead of $\{0,1\}^d$}
	

	
	
	\newcommand{\posR}{\mathbb{R}^{+}}
	
	\newcommand{\psirevx}{\psi^{x}_\textsf{rev}}
	\newcommand{\psirevy}{\psi^{y}_\textsf{rev}}

%% file: prelim.tex
\section{Preliminaries}
	
	In this paper, we use $\posR$ to denote the set of all positive reals. For notational convenience, we first give the formal definitions of the problem we study in this paper.
	
	\subsection{Problem List}\label{sec:problem-list}
	
	\begin{defi}[Boolean Vector Problem List]\label{defi:prob-list} For $n,d \in \mathbb{N}$, we define several problems. For all of them, the input is the same: we are given sets $A$ and $B$ of $n$ vectors from $\{0,1\}^d$. 		
		\begin{enumerate}
			\item $\OV_{n,d}$\footnote{Note that we consider the red-blue version of $\OV$ in this paper for convenience, and it is equivalent to the original monochromatic version.}: \emph{Given $A, B \subseteq \{0,1\}^d$ with $|A|=|B|=n$, determine whether there exists $(a,b) \in A \times B$ such that $a \cdot b = 0$.} 
			
			\item $\ExactIP_{n,d}$: \emph{Given $A, B$ as before, and an integer $0 \le m \le d$, determine whether there exists $(a,b) \in A \times B$ such that $a \cdot b = m$.}
			
			\item $\MaxIP_{n,d}$: \emph{Given $A, B$ as before, compute} \[\MAX(A,B) := \max_{a \in A, b\in B} a \cdot b.\]
			
			\item $\MinIP_{n,d}$: \emph{Given $A, B$ as before, compute} \[\MIN(A,B) := \min_{a \in A, b\in B} a \cdot b.\]
			
			\item $\MAMaxIP_{n,d}$: \emph{Given $A, B$ as before, output a number $\WMAX(A,B) \in [\MAX(A,B) /2, \MAX(A,B)]$.}
			
			
			\item $\MAMinIP_{n,d}$: \emph{Given $A, B$ as before, output a number $\WMIN(A,B) \in [\MIN(A,B),2 \cdot \MIN(A,B)]$.}
		\end{enumerate}
		
	\end{defi}

	\begin{rem}
	The constant factor $2$ in the definitions of $\MAMinIP$ and $\MAMaxIP$ is only chosen for convenience, it can be replaced by any constant $\kappa > 1$ (such as $1.001$, or $100$).
	\end{rem}

	\begin{defi}[Other Problems]\label{defi:prob-list-geo}
		We define the following problems.
		
		\begin{enumerate}
			\item $\BCPp[p]_n$: For a fixed real $p \in [1,2]$, given two sets $A,B$ of $n$ points in $\mathbb{R}^{d}$ where $d = n^{o(1)}$, compute $\min_{(a,b) \in A \times B} \|a - b\|_p$.
			
			\item $\FPp[p]_n$: For a fixed real $p \in [1,2]$, given a set $A$ of $n$ points in $\mathbb{R}^{d}$ where $d = n^{o(1)}$, compute $\max_{(a,b) \in A \times A} \|a - b\|_p$.
			
			\item $\JaccardIndexPair_n$: Given $A,B$ as two collections of $n$ sets of size $n^{o(1)}$, compute $\max_{(S,T) \in A \times B} J(S,T)$, where $J(S,T) := \frac{|S \cap T|}{|S \cup T|}$.
			
		\end{enumerate}
	\end{defi}
	
	\subsection{Locality-sensitive Hashing}\label{sec:LSH}
	
	In this paper we apply some well-known results from the theory of \emph{locality-sensitive hashing} (LSH) (See~\cite{wang2014hashing,andoni2018approximate} for excellent recent references on LSH families and their applications).

	\paragraph{$\ell_p$ Norm.} From the theory of $p$-stable distributions, LSH families for $\ell_p$ norm when $p \in [1,2]$ have been constructed.
	
	\begin{lemma}[\cite{DIIM04LSH}]\label{lm:LSH-ell-p}
		For a constant $p \in [1,2]$, the $\ell_p$ distance $\dist_p(x,y) := \|x-y\|_p$ admits a LSH family. Moreover, for all real $\eps \in (0,0.1)$ and real $R > 0$, $f^{\dist_p}_{R,(1+\eps) R}$ admits a $(p_1,p_2)$-sensitive LSH family, such that $p_1 - p_2 \ge \Omega(\eps)$.
	\end{lemma}

	\paragraph{Jaccard Index.} For two sets $A,B$, recall that their Jaccard index is defined as $J(A,B) := \frac{ |A \cap B|}{|A \cup B|}$. It is well-known that this measure admits a LSH family by the MinHash method.
	
	\begin{lemma}[\cite{broder1997resemblance}]\label{lm:LSH-Jaccard-Index}
		Let $0 \le p_2 < p_1 \le 1$ be two reals, and $f$ be the function on two sets such that $f(A,B) = 1$ when $J(A,B) \ge p_1$, $f(A,B) = 0$ when $J(A,B) \le p_2$ and undefined otherwise. $f$ admits a $(p_1,p_2)$-sensitive LSH family.
	\end{lemma}

%% file: Sigma2.tex
\section{General Reduction Frameworks with $\Sigma_2$ Communication Protocols and LSH Families}

In this section we present two general reduction frameworks for showing equivalence to $\OV$.

\subsection{$\Sigma_2$ Communication Protocols and Reductions to $\OV$}

We first show that an efficient $\Sigma_2^\cc$ protocol for a function $f$ implies a reduction from $f\SATPAIR$ to $\OV$.

\begin{reminder}{Theorem~\ref{theo:Sigma2-to-OV}}
	Let $F : \{0,1\}^{d} \times \{0,1\}^{d} \to \{0,1\}$ and $n \in \mathbb{N}$, suppose $F$ has a computationally-efficient $\Sigma_2^\cc$ protocol, in which Merlin sends $m_1$ bits, Megan sends $m_2$ bits, and Alice and Bob communicate $\ell$ bits. 		Then there is a reduction from every $F\SATPAIR_{n}$ instance $I$ to $\OV_{n,2^{(m_2+\ell)}}$ instances $J_1,J_2,\dotsc,J_{2^{m_1}}$, such that $I$ is a yes instance if and only if there is a $j$ such that $J_j$ is a yes instance. The reduction takes $n \cdot 2^{O(m_1 + m_2 + \ell)} \cdot \poly(d)$ time.
\end{reminder}

\begin{proofof}{Theorem~\ref{theo:Sigma2-to-OV}}
	\newcommand{\proofMerlin}{s_{\sf Merlin}}
	\newcommand{\proofMegan}{s_{Megan}}
	
	Let $F$ be the given function and $\Pi$ be its $\Sigma_2$ protocol. Fix $a \in \{0,1\}^{m_1}$ and $b \in \{0,1\}^{m_2}$ as the proofs from Merlin and Megan. Let $w_1,w_2,\dotsc,w_{2^\ell}$ be an enumeration of all possible communication transcripts between Alice and Bob (note they communicate $\ell$ bits). We define two binary vectors $R_x(a,b), R_y(a,b) \in \{0,1\}^{2^\ell}$ as follows: for all $a,b$, $R_x(a,b)_i = 1$ ($R_y(a,b)_i = 1$) if and only if the transcript $w_i$ is consistent with Alice's input $x$ (Bob's input $y$), and $w_i$ makes Alice reject. Note that since the transcript is uniquely determined by $x$, $y$, $a$ and $b$, only one $w_i$ is consistent with both $x$ and $y$ given the pair $(a,b)$. It follows that $ \langle R_x(a,b), R_y(a,b) \rangle = 0$ if and only if Alice accepts the pair $(a,b)$. 
	
	Now, suppose we are given an $F\SATPAIR_{n}$ instance $I$ with sets $A$ and $B$ of $n$ vectors from $\{0,1\}^{d}$. We first enumerate Merlin's possible string $a \in \{0,1\}^{m_1}$, and use $R_{x}(a,\cdot)$ to denote the string obtained by concatenating all $R_{x}(a,b)$'s for $b \in \{0,1\}^{m_2}$. $R_{y}(a,\cdot)$ is defined similarly. For each $a$, let $A_{a}$ be the set of $R_x(a,\cdot) \in \{0,1\}^{m_2 + \ell}$ for all $x \in A$, and $B_{a}$ be the set of $R_{y}(a,\cdot) \in \{0,1\}^{m_2 + \ell}$ for all $y \in B$.
	
	We claim $I$ is a yes instance if and only if some pair $(A_a,B_a)$ is a yes instance for $\OV$.
	
	\begin{itemize}
		\item Suppose $I$ is a yes instance. Then there is an $(x,y) \in A \times B$ such that $F(x,y) = 1$. By the definition of $\Sigma_2^\cc$ protocols and our constructions, there is an $a \in \{0,1\}^{m_1}$ such that for all $b \in \{0,1\}^{m_2}$ we have $\langle R_x(a,b), R_y(a,b)\rangle = 0$. Hence, for such an $a$, $\langle R_x(a,\cdot), R_y(a,\cdot)\rangle = 0$, and therefore $(A_a,B_a)$ is a yes instance for $\OV$.
		
		\item Suppose $I$ is a no instance. Then for all $(x,y) \in A \times B$, $F(x,y) = 0$. Hence, for all $a \in \{0,1\}^{m_1}$ and all $(x,y) \in A \times B$, we have $\langle R_x(a,\cdot), R_y(a,\cdot)\rangle \ne 0$, which means all $(A_a,B_a)$'s are no instances for $\OV$.
	\end{itemize}

	Finally, since $\Pi$ is computationally-efficient, the above reduction takes $O(n \cdot 2^{O(m_1 + m_2 + \ell)} \cdot \poly(d))$ time, which completes the proof.
\end{proofof}

\subsection{LSH Families and Reductions to Additive Approximate $\MaxIP$}

Next, we show that an efficient LSH family implies a reduction to additively approximating $\MaxIP$.

\begin{reminder}{Theorem~\ref{theo:LSH-to-MaxIP-general}}
	Suppose $f : \mathcal{X} \times \mathcal{X} \to \{0,1,\bot\}$ admits a $(p_1,p_2)$-sensitive LSH family. Let $\eps = p_1 - p_2$.
	
	Then there is a randomized reduction from $f\SATPAIR_{n}$ to computing an $\eps/8 \cdot d$ additive approximation to $\MaxIP_{n,d}$ with $d = O(\eps^{-2} \log n)$, which succeeds with probability at least $1 - 1/n$.		
\end{reminder}

\begin{proof}	
	Let $\mathcal{F}$ be the corresponding $(p_1,p_2)$-sensitive LSH family, and $\mathcal{S}$ be the co-domain for hash functions from $\mathcal{F}$. Consider the following process: draw $h$ from $\mathcal{F}$ uniformly at random, then map each item in $\mathcal{S}$ independently to the string $(0,1)$ or $(1,0)$, each with probability $0.5$. Let this map be $\varphi$. Composing $h$ and $\varphi$, we obtain a function $g(x) = \varphi(h(x))$ such that:
	\begin{itemize}
		\item If $f(x,y) = 1$, then $\langle g(x), g(y) \rangle = 1$ with probability at least $p_1 + (1-p_1) /2 \ge \frac{1}{2} + \frac{1}{2} \cdot p_1$.
		\item If $f(x,y) = 0$, then $\langle g(x), g(y) \rangle = 1$ with probability at most $p_2 + (1-p_2) /2  \le \frac{1}{2} + \frac{1}{2} \cdot p_2$.
	\end{itemize}
	
	Repeat the above process for $N = c \log n$ times, independently drawing functions $g_1,g_2,\dotsc,g_{N}$, where $c$ is a parameter to be specified later. We set our reduction $w(x)$ to be the concatenation of all $g_i(x)$'s. Let $\tau_1 = \frac{1}{2} + \frac{1}{2} \cdot (p_1 - \eps/4)$ and $\tau_2 = \frac{1}{2} + \frac{1}{2} \cdot (p_2 + \eps/4)$. By a simple Chernoff bound, there is a real $c_1 = \Theta(\eps^{2})$ such that
	\begin{itemize}
		\item If $f(x,y) = 1$, then $\langle w(x), w(y) \rangle > \tau_1 \cdot N $ with probability at least $1 - 2^{c_1 \cdot N}$.
		\item If $f(x,y) = 0$, then $\langle w(x), w(y) \rangle < \tau_2 \cdot N$ with probability at least $1 - 2^{c_1 \cdot N}$.
	\end{itemize}
	
	Set $c := 3/c_1$, and let $\Anew$ (respectively, $\Bnew$) be the set of $w(a)$'s for all $a \in A$ (the set of $w(b)$'s for all $b \in B$). It follows that with probability at least $1 - 1 / n$, if there is an $(x,y) \in A \times B$ with $f(x,y) = 1$ then $\MAX(\Anew,\Bnew) > \tau_1 \cdot N$, and if $f(x,y) = 0$ for all $(x,y) \in A \times B$, then $\MAX(\Anew,\Bnew) < \tau_2 \cdot N$. Observe this reduction satisfies the desired approximation property. 
\end{proof}

%% file: equiv-class.tex
\section{An Equivalence Class for Orthogonal Vectors}

	In this section we apply our two general frameworks to prove Theorem~\ref{theo:equiv-class-OV}. 

	\subsection{Equivalence Between Boolean Vectors Problem}	
	
	We first show that all Boolean vectors problems listed in Theorem~\ref{theo:equiv-class-OV} can be trivially reduced to $\ExactIP$, and $\OV$ can be reduced to all of them.
	
	\begin{lemma}\label{lm:trivial-or-known}
		The following holds:
		\begin{itemize}
			\item If $\ExactIP$ is in truly subquadratic time, then so are $\OV$, $\MAMinIP$ ($\MAMaxIP$) and $\MaxIP$ ($\MinIP$).
			\item If any of $\MAMinIP$ ($\MAMaxIP$), $\MaxIP$ ($\MinIP$) and $\ExactIP$ is in truly subquadratic time, then so is $\OV$.
		\end{itemize}
	\end{lemma}
	\begin{proof}
		For the first item,	$\MAMinIP$ ($\MAMaxIP$) and $\MaxIP$ ($\MinIP$) can all be trivially reduced to $\ExactIP$, and $\OV$ can be reduced to $\MaxIP$ by~\cite{Wil05}.
		
		For the second item, the case of $\MAMaxIP$ follows from Theorem~4.1 in~\cite{Rub18BetterMA}, and it is easy to see that $\OV$ can be trivially reduced to $\MinIP$ or $\MAMinIP$ ($\OV$ is equivalent to asking whether the minimum inner product is zero).
	\end{proof}
        
    
    Therefore, all we need is a reduction from $\ExactIP$ to $\OV$. We provide it by constructing a good $\Sigma_2$ communication protocol, and applying Theorem~\ref{theo:Sigma2-to-OV}.

	\begin{lemma}\label{lm:ExactIP-to-OV}
		If $\OV$ is in truly subquadratic time, then so is $\ExactIP$.
	\end{lemma}

	\begin{prop}\label{prop:Sigma2cc-for-IP}
		Let $\IP_{n,k} : \{0,1\}^{n} \times \{0,1\}^n \to \{0,1\}$ be the function that checks whether $\langle x,y \rangle = k$. For all $n, k \in \mathbb{Z}^{+}$, and a parameter $1 \le \ell \le n$, there is a $\Sigma_2^\cc$ computationally-efficient protocol for $\IP_{n,k}$ in which Merlin sends $\ell \cdot \lceil \log (\lceil n/\ell \rceil + 1) \rceil$ bits, Megan sends $\lceil \log \ell \rceil$ bits and Alice and Bob communicate $\lceil n/\ell \rceil$ bits.
	\end{prop}
	\begin{proof}
		We assume $\ell$ divides $n$ for simplicity. Let $x,y$ be the inputs of Alice and Bob, respectively. We partition $x$ into $\ell$ equally-sized groups of length $n/\ell$, let them be $x_1,x_2,\dotsc,x_\ell$. Similarly, we partition $y$ into groups $y_1,y_2,\dotsc,y_\ell$. Clearly, $\langle x,y \rangle = \sum_{i=1}^\ell \langle x_i,y_i \rangle$.
		
		Merlin's message is a vector $\psi \in \{0,1,\dotsc,n/\ell\}^{\ell}$, where $\psi_i$ is intended to be $\langle x_i,y_i \rangle$.
		
		Alice rejects immediately if $\sum_{i=1}^{\ell} \psi_i \ne k$, regardless of Megan's message. Otherwise, Megan's message is an index $i$ in $[\ell]$. Bob sends $y_i$ to Alice, and Alice accepts if and only if $\langle x_i,y_i \rangle = \psi_i$.		
		
		We argue the protocol correctly decides $\IP_{n,k}$.
		If $\langle x,y \rangle = k$, it is easy to see that for the correct $\psi$, Alice accepts all messages from Megan (and Bob). When $\langle x,y \rangle \ne k$, for all $\psi$ such that $\sum_{i=1}^{\ell} \psi_i = k$ (otherwise Alice always rejects), there must be an $i$ such that $\langle x_i,y_i \rangle \ne \psi_i$, which means Alice rejects on the pair $\psi$ and $i$. Finally, it is easy to see that the protocol satisfies the requirements of computational efficiency, which completes the proof.
	\end{proof}

	Now we are ready to prove Lemma~\ref{lm:ExactIP-to-OV}.
    		
	\begin{proofof}{Lemma~\ref{lm:ExactIP-to-OV}}
				
		Suppose there is a universal constant $\delta > 0$ such that for all constants $c'$, $\OV_{n,c'\log n}$ can be solved in $n^{2 - \delta}$ time. Let $c$ be an arbitrary constant.
		
		Observe that an $\ExactIP_{n, c\log n}$ instance with target integer $m$, is simply a $\IP_{c\log n,m}\SATPAIR_n$ instance. Set $\ell := \eps \cdot \log n$ for an $\eps > 0$ to be specified later. By Proposition~\ref{prop:Sigma2cc-for-IP}, there is a $\Sigma_2^\cc$ protocol for $\IP_{c \log n,m}$ such that Merlin sends $\eps \cdot \log (c/\eps) \cdot \log n$ bits, Megan sends $\log (\epsilon \log n)$ bits and Alice and Bob communicate $c/\epsilon$ bits.
		
		By Theorem~\ref{theo:Sigma2-to-OV}, there is a reduction from an $\ExactIP_{n, c\log n}$ instance to $2^{\eps \log (c/\eps) \log n} = n^{\eps \log(c/\eps)}$ many $\OV_{n,O(2^{c/\eps}\log n)}$ instances. We can set $\eps$ so that $\eps \log(c/\eps) < \delta / 2$. Note that $\eps$ only depends on $c$ and $\delta$, so it is still a fixed constant, which means (by assumption) that $\OV_{n,O(2^{c/\eps}\log n)}$ can be solved in $n^{2 - \delta}$ time. Applying the algorithm for $\OV$, we get an $n^{2 - \delta/2}$ time algorithm for $\ExactIP_{n, c\log n}$, which completes the proof. 
	\end{proofof}

	\subsection{Equivalences Between $\OV$ and Approximation Problems}
	
	Now we deal with approximation problems in Theorem~\ref{theo:equiv-class-OV}.
	
	\subsubsection*{\BCPp[p] and \FPp[p]} 
	
	We first show $\OV$ is equivalent to approximate $\BCPp[p]$, $\FPp[p]$ and additive approximate $\MaxIP$. One direction is already established in~\cite{Rub18BetterMA}.
	
	\begin{lemma}[Theorem 4.1 of~\cite{Rub18BetterMA}]\label{lm:OV-to-BCPp-FPp}
		If \BCPp[p] or $\FPp[p]$ can be approximated in truly subquadratic time for any $p \in [1,2]$ or $\MaxIP$ can be additively approximated in truly subquadratic time, then $\OV$ is in truly subquadratic time.\footnote{\cite{Rub18BetterMA} only discussed $\BCPp[p]$ and additive approximation to $\MaxIP$, but it is easy to see that the proof also works for $\FPp[p]$.}
	\end{lemma}

	In the following we show the reverse also holds.
	
	\begin{lemma}\label{lm:BCPp-FPp-to-OV}
		If $\OV$ is in truly-subquadratic time, then for all $p \in [1,2]$, $\BCPp[p]$ and $\FPp[p]$ can be approximated in truly subquadratic time, and $\MaxIP$ can be additively approximated in truly subquadratic time.
	\end{lemma}

	We are going to apply Theorem~\ref{theo:LSH-to-MaxIP-general} and will actually prove a much stronger result. We show that for any metric $\dist : \mathcal{X} \times \mathcal{X} \to \mathbb{R}_{\ge 0}$ which admits a \emph{Locality-sensitive hashing} (LSH) family, approximate $\BCP$ and $\FP$ with respect to $\dist$ can be efficiently reduced to $\OV$.
	
	In the following, we use $\BCP_{n,\dist}$ and $\FP_{n,\dist}$ to denote the corresponding problems with respect to the metric $\dist$. Now we are ready to give the reduction.

	\begin{reminder}{Corollary~\ref{cor:LSH-to-MaxIP-metric}}
		For a distance function $\dist : \mathcal{X} \times \mathcal{X} \to \mathbb{R}_{\ge 0}$ which admits an LSH family, $\BCP_{n,\dist}$ and $\FP_{n,\dist}$ can be approximated in truly subquadratic time if $\OV$ is in truly subquadratic time.
	\end{reminder}
	\begin{proof}
		Suppose $\OV$ is in truly subquadratic time. By Lemma~\ref{lm:trivial-or-known} and Lemma~\ref{lm:ExactIP-to-OV}, $\MaxIP$ and $\MinIP$ are also in truly-subquadratic time. In the following we only discuss $\BCP_{n,\dist}$; the reduction for $\FP_{n,\dist}$ is analogous (with $\MinIP$ in place of $\MaxIP$).
		
		Let $\eps > 0$ be an arbitrary constant. We want to approximate the minimum distance between two sets $A$ and $B$ of $n$ elements from $\mathcal{X}$ within a $(1+\epsilon)$ multiplicative factor. By a standard (simple) search to decision reduction that incurs only a negligible factor in the running time, we only have to consider the decision version, in which you are given a real $R$, and want to distinguish the following two cases: (1) $\min_{(a,b) \in A \times B} d(a,b) \le R$; (2) $\min_{(a,b) \in A \times B} d(a,b) \ge (1+\eps) \cdot R$. 
		
		By Theorem~\ref{theo:LSH-to-MaxIP-general}, this decision problem can be reduced to additive approximation to $\MaxIP_{n,O(\log n)}$, which is in truly-subquadratic time by Lemma~\ref{lm:ExactIP-to-OV}. This completes the proof.
	\end{proof}

	Now, from the LSH families for $\ell_p$-metric, Lemma~\ref{lm:BCPp-FPp-to-OV} follows directly.
 
	\begin{proofof}{Lemma~\ref{lm:BCPp-FPp-to-OV}}
		Assume $\OV$ is in truly-subquadratic time. It follows directly from Corollary~\ref{cor:LSH-to-MaxIP-metric} and Lemma~\ref{lm:LSH-ell-p} that for all $p \in [1,2]$, $\BCPp[p]$ and $\FPp[p]$ can be approximated in truly subquadratic time. 
		
		Also, by a simple random sampling method and a Chernoff bound (see e.g. Lemma 3.6 of~\cite{chen2018hardness}), computing an $\epsilon \cdot d$ additive approximation to $\MaxIP_{n,d}$ can be reduced to $\MaxIP_{n,O(\eps^{-2} \log n)}$, which can be solved in truly-subquadratic time by Lemma~\ref{lm:ExactIP-to-OV} and Lemma~\ref{lm:trivial-or-known}.
	\end{proofof}

	\subsection*{$\JaccardIndexPair$}
	Finally, we show the equivalence between $\OV$ and approximate $\JaccardIndexPair$.
	
	\begin{lemma}\label{lm:jidx-and-OV}
		$\OV$ is in truly-subquadratic time if and only if $\JaccardIndexPair$ can be additively approximated in truly-subquadratic time.
	\end{lemma}
	\begin{proof}
		For one direction, suppose $\OV$ is in truly subquadratic time. Using a similar argument as in Corollary~\ref{cor:LSH-to-MaxIP-metric}, from Lemma~\ref{lm:LSH-Jaccard-Index} and Theorem~\ref{theo:LSH-to-MaxIP-general} it follows that $\JaccardIndexPair$ can be additively approximated in truly-subquadratic time.
		
		For the other direction, suppose $\JaccardIndexPair$ can be additively approximated in truly subquadratic time. By Lemma~\ref{lm:OV-to-BCPp-FPp}, it suffices to show that $\MaxIP$ can be additively approximated in truly-subquadratic time. Given a $\MaxIP_{n,d}$ instance with sets $A,B$ consisting of $n$ vectors from $\{0,1\}^d$, suppose we want to compute an $\epsilon \cdot d$ approximation to it. In the following we show how to reduce it to a $\JaccardIndexPair$ instance.
		
		We begin by setting up some notation. For $t \in [d]$, we use $e^{[t]}$ to denote the Boolean vector $1^{t}0^{d-t}$ from $\{0,1\}^d$ (that is, the first $t$ coordinates are $1$, and the rest are $0$). For two vectors $a,b$, we use $a \circ b$ to denote their concatenation.
		
		\newcommand{\hx}{\widehat{x}}
		\newcommand{\hy}{\widehat{y}}
		
		For each $x \in A \subseteq \{0,1\}^d$ and $y \in B\subseteq \{0,1\}^d$, we create two vectors $\hx,\hy \in \{0,1\}^{3d}$, as follows:
		\[
		\hx = x \circ e^{[d - \|x\|_1]} \circ e^{[0]}, \hy = y \circ e^{[0]} \circ e^{[d - \|y\|_1]}.
		\]
		
		Interpreting $\hx$ and $\hy$ as indicator vectors, we create their corresponding sets $S_x,T_y \subseteq [3d]$. That is, for $i \in [3d]$, $\hx_i = 1$ if and only if $i \in S_x$ (the same holds for $\hy$ and $T_y$). Observe that
		\begin{equation}
		J(S_x,T_y) = \frac{|S_x \cap T_y|}{|S_x \cup T_y|} = \frac{\langle x,y \rangle}{2 d - \langle x,y \rangle}.\label{eq:JSxTy}
		\end{equation}
		
		\newcommand{\hA}{\widehat{A}}
		\newcommand{\hB}{\widehat{B}}
		
		Now we create $\hA$ and $\hB$ as the sets of all $S_x$ for $x \in A$ and $T_y$ for $y \in B$. Let $t = \max_{(S,T) \in \hA \times \hB} J(S,T)$ and $w = \max_{(a,b) \in A \times B} \langle a, b \rangle$. From Equation~\eqref{eq:JSxTy}, we can see $t = \frac{w}{2d - w}$ and $w = d \cdot 2 \cdot \frac{t}{t+1}$. Therefore, an $\eps/3$ approximation to $t$ is enough to obtain an $\eps \cdot d$ approximation to $w$, which completes the reduction.
	\end{proof}

	 And Theorem~\ref{theo:equiv-class-OV} follows from Lemma~\ref{lm:trivial-or-known}, Lemma~\ref{lm:ExactIP-to-OV}, Lemma~\ref{lm:OV-to-BCPp-FPp}, Lemma~\ref{lm:BCPp-FPp-to-OV} and Lemma~\ref{lm:jidx-and-OV}.
	 

%% file: MDOV.tex
\section{Equivalences for Moderate Dimensional Problems}

In this section we prove our equivalence theorems for moderate dimensional Boolean vectors problems.

\subsection{$\OV$ and $\MAMinIP$}

We first show moderate dimensional $\OV$ and $\MAMinIP$ are equivalent.

\begin{reminder}{Theorem~\ref{theo:equiv-OV-MAMinIP-MD}}
	Moderate dimensional $\OV$ is in truly subquadratic time if and only if moderate dimensional $\MAMinIP$ is.
\end{reminder}

To prove Theorem~\ref{theo:equiv-OV-MAMinIP-MD}, we construct the following reduction.

\begin{lemma}\label{lm:reduction-MAMinIP-to-OV}
	For all integers $n,d$ and a parameter $\eps > 0$, an $\MAMinIP_{n,d}$ instance can be reduced to $n^{O(\eps)}$ $\OV_{n,d^{O(1/\eps)} \log n}$ instances. The reduction is randomized and succeeds with probability at least $2/3$, and it takes $n^{1+O(\eps)} \cdot d^{O(1/\eps)}$ time.
\end{lemma}

Before proving Lemma~\ref{lm:reduction-MAMinIP-to-OV}, we show it implies Theorem~\ref{theo:equiv-OV-MAMinIP-MD}.

\begin{proofof}{Theorem~\ref{theo:equiv-OV-MAMinIP-MD}} 		Recall that $\MIN(A,B) := \min_{(a,b) \in A \times B} \langle a,b \rangle$.
	For the first direction, note that $\OV$ with two sets $A$ and $B$ essentially asks whether $\MIN(A,B) = 0$, and a $2$-approximation to $\MIN(A,B)$ is already enough to answer that question. Therefore, if moderate dimensional $\MAMinIP$ is in truly subquadratic time, then so is $\OV$.
	
	For the second direction, suppose there are constants $\eps_1,\delta_1 > 0$ such that $\OV_{n,n^{\delta_1}}$ can be solved in $n^{2-\eps_1}$ time. Let $\eps$ be a parameter to be set later, by Lemma~\ref{lm:reduction-MAMinIP-to-OV}, there are constants $c_1,c_2$ such that all $\MAMinIP_{n,n^\delta}$ instance can be efficiently reduced to $n^{c_1 \eps}$ $\OV_{n,n^{\delta c_2/\eps}}$ instances.
	
	We set $\eps$ such that $c_1 \eps = \eps_1 / 2$, and $\delta$ such that $\delta \cdot c_2 / \eps < \delta_1$. Then applying the algorithm for $\OV$, $\MAMinIP_{n,n^\delta}$ can be solved in $n^{2 - \eps_1 / 2}$ time, which completes the proof.
\end{proofof}

The following probability inequality will be useful in the proof of Lemma~\ref{lm:reduction-MAMinIP-to-OV}.
\begin{lemma}\label{lm:bound}
	Letting $\eps \in (0,0.1)$, and $\distr$ be a distribution on $\{0,1\}$ such that $\Ex_{X \sim \distr}[X] = \eps$, there is a universal constant $c$ such that for any integer $m$ and any $cm$ independent random variables $X_1,X_2,\dotsc,X_{cm}$ from $\distr$, we have
	\[
	\Pr\left[ \sum_{i=1}^{cm} X_i \ge \frac{1}{2} \cdot cm\right] \le \eps^{-m}.
	\]
\end{lemma}
The proof of Lemma~\ref{lm:bound} can be found in the appendix.

Finally, we prove Lemma~\ref{lm:reduction-MAMinIP-to-OV}.
\begin{proofof}{Lemma~\ref{lm:reduction-MAMinIP-to-OV}}
	
	Before presenting the reduction, we first introduce some notation. For a vector $x \in \{0,1\}^d$, and a subset $S \subset [d]$, $x_{|S} \in \{0,1\}^{|S|}$ denotes the projection of $x$ onto the coordinates of $S$. Similarly, for a sequence $T$ of integers from $[d]$, let $x_{|T} \in \{0,1\}^{|T|}$ denote the projection of $x$ on $T$, such that $\left( x_{|T} \right)_i := x_{T_i}$ for each $i \in [|T|]$. We also use the Iverson bracket notation: for a predicate $P$, $\left[ P \right]$ takes value $1$ when $P$ is true, and $0$ otherwise.
	
	\paragraph*{Reduction to a Decision Problem.} Our reduction will focus on a corresponding decision problem: given two sets $A,B$ of $n$ vectors from $\{0,1\}^d$ and an integer $\tau \le d/2$, we want to distinguish the following two cases: $\MIN(A,B) \ge 2\tau$ or $\MIN(A,B) \le \tau$ (the algorithm can output anything when $\tau < \MIN(A,B) < 2\tau$). It is easy to see that via a binary search, $\log d$ calls to this decision problem can be used to solve the original $\MAMinIP$ problem, and a factor of $\log d \le \log n$ can be ignored here.
	
	\paragraph*{One Step Reduction with $\distr_T$.} Now, suppose we pick a sequence of $d/\tau$ uniform random numbers from $[d]$ and let $\distr_T$ be its distribution. Then for $x,y \in \{0,1\}^d$, we have:
	
	\begin{itemize}
		\item If $\langle x, y \rangle \le \tau$:
		\begin{align*}
		\Pr_{T \leftarrow \distr_T}[\langle x_{|T}, y_{|T} \rangle = 0] \ge (1 - \tau/d)^{d/\tau} &\ge \left(1-\frac{1}{2}\right)^{2}\\ 
				&> 0.25.
		\end{align*}
		\item If $\langle x, y \rangle \ge 2\tau$:
		\[
		\Pr_{T \leftarrow \distr_T}[\langle x_{|T}, y_{|T} \rangle = 0] \le (1 - 2\tau/d)^{d/\tau} \le e^{-2} < 0.14.
		\]
	\end{itemize}
	
	The important observation is that there is a constant probability gap between the above two cases.	
	
	\paragraph*{A Micro Reduction to $\OV$.} Now, let $N$ be an integer and $\distr_T^{\otimes N}$ be the joint distribution of $N$ independent samples from $\distr_T$. We write $\{T_i\} \leftarrow \distr_T^{\otimes N}$ to denote that $(T_1,T_2,\dotsc,T_N)$ is a random sample from $\distr_T^{\otimes N}$. By a standard Chernoff bound, when $\{T_i\} \leftarrow \distr_T^{\otimes N}$, there is a constant $c_1$ such that:
	
	\begin{itemize}
		\item If $\langle x, y \rangle \le \tau$:
		\[
		\Pr \left[ \sum_{i=1}^{N} \left[\langle x_{|T_i}, y_{|T_i} \rangle = 0 \right] > 0.2N \right] \ge 1 - 2^{-c_1 N}.
		\]
		\item If $\langle x, y \rangle \ge 2\tau$:
		\[
		\Pr \left[ \sum_{i=1}^{N} \left[\langle x_{|T_i}, y_{|T_i} \rangle = 0 \right] < 0.2N \right] \ge 1 - 2^{-c_1 N}.
		\]
	\end{itemize}
	
	Now, for a fixed $\{T_i\}$, we can distinguish the above two cases via a reduction to a ``micro'' $\OV$ instance.
	
	Note that $\sum_{i=1}^{N} \left[\langle x_{|T_i}, y_{|T_i} \rangle = 0 \right] > 0.2N$ is equivalent to the condition that there is are $t = 0.8N$ pairs $(i_1,j_1),(i_2,j_2),\dotsc,(i_t,j_t) \in [N] \times [d/\tau]$ such that all $i_k$'s are distinct, and for all $k \in [t]$, $\left(x_{|T_{i_k}}\right)_{j_k} \cdot \left(y_{|T_{i_k}}\right)_{j_k} = 1$.
	
	With this observation, we can construct our reduction. There are 
	\[
	L = \binom{N}{t} \cdot (d/\tau)^{t} = (d/\tau)^{O(N)}
	\]
	possible $t$-tuples of pairs. We sort them in an arbitrary but consistent order. Now we construct a mapping $\phi_{\{T_i\}} : \{0,1\}^{d} \to \{0,1\}^L$ as follows:
	
	For each $\ell \in [L]$, let $(i_1,j_1),(i_2,j_2),\dotsc,(i_t,j_t)$ be the $\ell$-th $t$-tuple of pairs. For a vector $z \in \{0,1\}^{d}$, we set $\phi_{\{T_i\}}(z)_\ell = 1$, iff $\left(z_{|T_{i_k}}\right)_{j_k} = 1$ for all $k \in [t]$.
	
	Then for all $x,y \in \{0,1\}^d$, we have $\sum_{i=1}^{N} \left[\langle x_{|T_i}, y_{|T_i} \rangle = 0 \right] > 0.2N$ is further equivalent to $\langle \phi_{\{T_i\}}(x), \phi_{\{T_i\}}(y) \rangle = 0$. For convenience, we let $\distr_{\phi}$ denote the distribution of $\phi_{\{T_i\}}$ when $\{T_i\}$ is drawn from $\distr_T^{\otimes N}$ and we set $N = \eps^{-1} / c_1$.
	
	To summarize, we have:
	
	\begin{itemize}
		\item If $\langle x, y \rangle \le \tau$:
		\[
		\Pr_{\phi \leftarrow \distr_\phi} \left[ \langle \phi(x), \phi(y) \rangle = 0 \right] \ge 1 - 2^{-\eps^{-1}}.
		\]
		\item If $\langle x, y \rangle \ge 2\tau$:
		\[
		\Pr_{\phi \leftarrow \distr_\phi} \left[ \langle \phi(x), \phi(y) \rangle > 0 \right] \ge 1 - 2^{-\eps^{-1}}.
		\]
	\end{itemize}
	
	\paragraph*{The Final Reduction.} Finally, letting $c_2$ be the universal constant in Lemma~\ref{lm:bound}, we pick $m = 3 c_2 \cdot \eps \log n$ i.i.d. mappings $\phi_1,\phi_2,\dotsc,\phi_m$ from $\distr_\phi$. Applying Lemma~\ref{lm:bound}, we have:
	
	\begin{itemize}
		\item If $\langle x, y \rangle \le \tau$:
		\[
		\Pr_{\{ \phi_i \} \leftarrow \distr_{\phi}^{\otimes m}} \left[ \sum_{i=1}^{m} \left[\langle \phi_i(x), \phi_i(y) \rangle = 0 \right] > \frac{1}{2} \cdot m \right] \ge 1 - n^{-3}.
		\]
		\item If $\langle x, y \rangle \ge 2\tau$:
		\[
		\Pr_{\{ \phi_i \} \leftarrow \distr_{\phi}^{\otimes m}} \left[ \sum_{i=1}^{m} \left[\langle \phi_i(x), \phi_i(y) \rangle = 0 \right] < \frac{1}{2} \cdot m \right] \ge 1 - n^{-3}.
		\]
	\end{itemize}
	
	Now, we use our final reduction to distinguish the above two cases. Note that $\sum_{i=1}^{m} \left[\langle \phi_i(x), \phi_i(y) \rangle = 0 \right] > \frac{1}{2} \cdot m$ is equivalent to the condition that there is a subset $S \subseteq [m]$ with $|S| > \frac{1}{2} \cdot m$ such that $\langle \phi_i(x), \phi_i(y) \rangle = 0$ for all $i \in S$.
	
	We enumerate all possible such subsets $S$. For a vector $z \in \{0,1\}^{d}$, we define $\phi_S(z)$ to be the concatenation of $\phi_i(z)$'s for all $i \in S$. We set $A_S$ as the set of all $\phi_S(x)$'s for $x \in A$, and $B_S$ as the set of all $\phi_S(y)$'s for $y \in B$.
	
	Then we can see that $\sum_{i=1}^{m} \left[\langle \phi_i(x), \phi_i(y) \rangle = 0 \right] > \frac{1}{2} \cdot m$ is further equivalent to whether there is a subset $S$ with $|S| > \frac{1}{2} \cdot m$ and $(A_S,B_S)$ is a yes instance for $\OV$.
	
	\paragraph*{Summary.} Putting everything together, we have a randomized reduction to $T = 2^{O(\eps \log n)} = n^{O(\eps)}$ $OV_{n, (d/\tau)^{O(1/\eps)} \log n}$ instances with set-pairs $(A_1,B_1),(A_2,B_2),\dotsc,(A_T,B_T)$ such that, with probability at least $1 - 1/n$:
	
	\begin{itemize}
		\item If $\MIN(A,B) \le \tau$, then one of the $(A_i,B_i)$ is a yes instance for $\OV$.
		\item If $\MIN(A,B) \ge 2\tau$, all $(A_i,B_i)$'s are no instance for $\OV$.
	\end{itemize}
	
	The above completes the proof.
\end{proofof}

\subsection{$\ExactIP$, $\MaxIP$ and $\MinIP$} 

Now we proceed to show moderate dimensional $\ExactIP$, $\MaxIP$ and $\MinIP$ are equivalent.

\begin{reminder}{Theorem~\ref{theo:equiv-MaxIP-MinIP-ExactIP-MD}}
	For moderate dimensional $\MaxIP$, $\MinIP$ and $\ExactIP$, either all of them are in truly subquadratic time, or none of them are.
\end{reminder}

To prove the above theorem, we need the following two simple reductions, whose proofs can be found in the appendix.

\begin{lemma}\label{lm:MinIP-to-MaxIP}
	There are functions $\psirevx,\psirevy : \{0,1\}^* \to \{0,1\}^*$ such that for all integer $d$ and $x,y \in \{0,1\}^d$, 
	we have $\psirevx(x),\psirevy(y) \in \{0,1\}^{2d}$ and $\langle \psirevx(x), \psirevy(y) \rangle = d - \langle x, y \rangle$.
\end{lemma}

\begin{lemma}\label{lm:ExactIP-to-MinIP}
	For all integers $d$ and $0 \le m \le d$, there are mappings $\varphi^x_{d,m},\varphi^y_{d,m} : \{0,1\}^{d} \to \{0,1\}^{O(d^2)}$ and an integer $M_{d}$, such that for all $x,y \in \{0,1\}^d$:
	
	\begin{itemize}
		\item If $\langle x, y \rangle = m$, then $\langle \varphi^x_{d,m}(x), \varphi^y_{d,m}(y) \rangle = M_{d}$.
		\item Otherwise, $\langle \varphi^x_{d,m}(x), \varphi^y_{d,m}(y) \rangle > M_{d}$.
	\end{itemize}
	
\end{lemma} 

\begin{proofof}{Theorem~\ref{theo:equiv-MaxIP-MinIP-ExactIP-MD}}
	By Lemma~\ref{lm:MinIP-to-MaxIP}, one can easily reduce a $\MaxIP_{n,d}$ instance to a $\MinIP_{n,2d}$ and vice versa. Therefore, moderate dimensional $\MaxIP$ and $\MinIP$ are truly-subquadratic equivalent. We only need to show that moderate dimensional $\MinIP$ and $\ExactIP$ are equivalent.
	
	Assuming moderate dimensional $\ExactIP$ is in truly subquadratic time, so there are two constants $\eps$ and $\delta$ such that $\ExactIP_{n,n^\delta}$ can be solved in $n^{2-\eps}$ time. Let $\delta' = \min(\eps,\delta) / 2$. Given a $\MinIP_{n,n^{\delta'}}$ instance, by enumerating all possible inner products between $0$ and $n^{\delta'}$, we can reduce the instance to $n^{\delta'}$ instances of $\ExactIP_{n,n^{\delta'}}$. Applying the algorithm for $\ExactIP$, we then have an $n^{2-\varepsilon + \delta'} \le n^{2-\delta'}$ time algorithm for $\MinIP_{n,n^{\delta'}}$. Hence, moderate dimensional $\MinIP$ is also in truly-subquadratic time.
	
	Finally, assume moderate dimensional $\MinIP$ is in truly subquadratic time. Note that by Lemma~\ref{lm:ExactIP-to-MinIP}, an $\ExactIP_{n,d}$ instance can be reduced to a $\MinIP_{n,O(d^2)}$ instance, which immediately implies that moderate dimensional $\ExactIP$ is also in truly subquadratic time.
\end{proofof}

%% file: tighter.tex
\section{Tighter Connection Between $\MaxIP$, $\BCPp$ and \FPp}

In this section we establish the tighter connections between $\MaxIP$, $\BCPp$ and $\FPp$. 

In Section~\ref{sec:additive-Max-IP-Exact-IP}, we show tighter connections for $\MaxIP$, $\ExactIP$ and additive approximation to $\MaxIP$. And in Section~\ref{sec:additive-Max-IP-other-geometry}, we show similar connections for additive approximation to $\MaxIP$, $\BCPp$ and $\FPp$.

\subsection{Tighter Connection between $\ExactIP$, $\MaxIP$ and Additive Approximation to $\MaxIP$}\label{sec:additive-Max-IP-Exact-IP}

The following lemma is implicit in~\cite{Rub18BetterMA}, which is used to show $\BCPp[p]$ can not be approximated in truly-subquadratic time under $\SETH$. \cite{Rub18BetterMA} only states a reduction from $\OV$. However, the $\MA$ protocol in~\cite{Rub18BetterMA} works equally well for the Inner Product problem, so it actually gives a reduction from $\ExactIP$.

\begin{lemma}[Implicit in Theorem~4.1 of~\cite{Rub18BetterMA}]\label{lm:ExactIP-to-additive-MaxIP-Rub18}
	For all sufficiently large integers $n,c$ and a parameter $\eps > 0$, an $\ExactIP_{n,c\log n}$ instance can be reduced to $n^{O(\eps \log (c/\eps))}$ instances of computing $\Omega(1/\exp\{\WO(c/\epsilon)\}) \cdot d$ additive approximation to $\MaxIP_{n,d}$ for $d = n^{o(1)}$.
\end{lemma}

In order to prove our tighter connection, our goal here is to improve the additive approximation ratio from $\Omega(1/\exp\{\WO(c/\epsilon)\} )$ to $\Omega(1/\poly(c/\epsilon))$.

\subsubsection{A New $\MA$ Protocol for Inner Product}

\newcommand{\Piorig}{\Pi_{\sf orig}}
\newcommand{\Pinew}{\Pi_{\sf new}}

For that purpose, we need to modify the $\MA$ protocol from~\cite{Rub18BetterMA}. In the following, we first describe the $\MA$ protocol for Inner Product in~\cite{Rub18BetterMA} based on AG codes. Below we only summarize the relevant properties we need; readers can refer to~\cite{Rub18BetterMA} for the details of the protocol.

\newcommand{\veca}{\vec{a}}
\newcommand{\vecb}{\vec{b}}

\newcommand{\fin}{\mathbb{F}}
\newcommand{\fqtwo}{\fin_{q^2}}

\begin{lemma}[Theorem 3.1~\cite{Rub18BetterMA}]\label{lm:MA-protocol-Rub18}
	For every $T \in [2,N]$, there is a computationally-efficient $\MA$ protocol for Inner Product such that
	
	\begin{enumerate}
		\item Alice and Bob hold input $x,y \in \{0,1\}^{N}$ respectively, and want to decide whether $\langle x,y \rangle = m$ for a target integer $m$.
		
		\item Set $q$ to be the first prime larger than $T$ and a universal constant $c_1$, and set $R = \log(N/T) + O(1)$.
				
		\item Merlin sends Alice a vector $z \in \fin_{q^2}^{2^R}$, Alice rejects $z$ immediately if it doesn't satisfy some conditions.
		
		\item Alice and Bob then toss $R$ coins to get $r \in [2^R]$. Based on $x$ (or $y$) and $r$, Alice and Bob generate two vectors in $\mathbb{F}_{q^2}^{T}$, $\veca(x,r)$ and $\vecb(y,r)$ respectively, 
		
		\item Bob sends Alice $\vecb(y,r)$, and Alice calculates $u(x,y,r) =  \langle \veca(x,r), \vecb(y,r) \rangle$. Alice accepts if and only if $u(x,y,r) = z_r$.
	\end{enumerate}

	The protocol satisfies the following conditions:
	
	\begin{itemize}
		\item If $\langle x,y \rangle = m$, then there is a proof (the vector $z$) from Merlin such that Alice always accepts.
		
		\item If $\langle x,y \rangle \ne m$, then for all proofs from Merlin, Alice accepts with probability at most $1/2$.
	\end{itemize}
\end{lemma}

\newcommand{\Pirred}{P_{\textsf{irred}}}

\paragraph*{Our Modified Protocol.} We make some minor modifications to the above protocol. First, note that an element from $\fqtwo$ can be treated as an element in $\fin_q[x] / (\Pirred(x))$, where $\Pirred(x) \in \fin_q[x]$ is an irreducible polynomial of degree $2$. In this way, we can interpret all elements in $\veca(x,r)$ and $\vecb(y,r)$ as degree $1$ polynomials in $\fin_q[x]$, which can in turn be interpreted as degree $1$ polynomials in $\mathbb{Z}[x]$. We denote these vectors of polynomials by $\vec{U}(x,r),\vec{V}(y,r) \in \mathbb{Z}[x]^{T}$, with coefficients from $\{0,1,\dotsc,q - 1\}$.

Next, we set $W(x,y,r) = \langle \vec{U}(x,r),\vec{V}(y,r) \rangle$, which is a degree $2$ polynomial in $\mathbb{Z}[x]$. Note that the coefficients of $W(x,y,r)$ are between $0$ and $O(q^2 \cdot T) = O(T^3)$.

Now, in the message from Merlin, for all possible $r \in [2^R]$, we also add a claimed description of $W(x,y,r)$. This takes $O\left(\frac{N \log T}{T}\right)$ bits, so it doesn't affect the message complexity from Merlin. Then, after Alice receives $\vecb(y,r)$ from Bob (from which she can obtain $\vec{V}(y,r)$), Alice computes $W(x,y,r)$ instead of $u(x,y,r)$, and rejects immediately if this $W(x,y,r)$ does not match the one given by Merlin. After that, she knows that $u(x,y,r) = W(x,y,r) / (\Pirred(x))$, and proceeds as in the original protocol.

It is easy to see that, when $\langle x,y \rangle = m$, if Merlin provides the correct $W(x,y,r)$'s, then Alice still always accepts (regardless of $r$). And when $\langle x,y \rangle \ne m$, since these $W(x,y,r)$'s only provide additional checks, Alice still accepts with probability at most $1/2$ for all proofs.

We use $\Piorig$ to denote the protocol from~\cite{Rub18BetterMA} (Lemma~\ref{lm:MA-protocol-Rub18}), and $\Pinew$ to denote our new protocol. In the following we utilize $\Pinew$ to give an improved reduction from $\ExactIP$ to additive approximation to $\MaxIP$.

Before that, we need the following encoding trick, whose proof can be found in the appendix.

\begin{lemma}\label{lm:encoding-trick}
	For all integers $d,r$ and $0 \le m \le dr^2$, there are mappings $\varphi^x,\varphi^y : \{0,1,\dotsc,r\}^{d} \to \{0,1\}^{O(d r^2 )^2}$ and an integer $0 \le M \le O(dr^2)^2$, such that for all $x,y \in \{0,1,\dotsc,r\}^{d}$:
	
	\begin{itemize}
		\item If $\langle x,y \rangle = m$, then $\langle \varphi^x(x),\varphi^y(y) \rangle = M$.
		\item Otherwise, $\langle \varphi^x(x),\varphi^y(y) \rangle < M$.
		\item Moreover, $M$ only depends on $d$ and $r$.
	\end{itemize}
	
\end{lemma} 

\begin{lemma}\label{lm:ExactIP-to-additive-MaxIP}
	For all sufficiently large integers $n,c$ and a parameter $\eps > 0$, every $\ExactIP_{n,c\log n}$ instance can be reduced to $n^{O(\eps \log (c/\eps))}$ instances of computing an $\Omega((\eps/c)^6) \cdot d$ additive approximation to $\MaxIP_{n,d}$ for $d = n^{o(1)}$.
\end{lemma}

\begin{proof}
	Consider an $\ExactIP_{n,c\log n}$ instance with sets $A$ and $B$, and integer $m$. Using our protocol $\Pinew$ for checking whether $\langle x,y \rangle = m$, we only need to figure out whether there is a pair $(x,y) \in A \times B$ and a proof from Merlin such that Alice always accepts.
	
	Let $N = c\log n$, and set $T = c/\eps$. Then the message complexity from Merlin is $O(\eps \log n \log(c/\eps))$ and the total number of random bits is $R = \log (N/T) + O(1) \le \log (\epsilon \log n) + O(1)$.
	
	We first enumerate all valid proofs $\psi$, which is a pair of $z \in \fin_{q^2}^{2^R}$ and $W \in \mathbb{Z}[x]^{2^R}$ such that for all $r \in [2^R]$, we have $z_r = W_{r} / \Pirred(x)$.
	
	Next, we want to determine whether there is a pair $(x,y) \in A \times B$, such that this proof $\psi$ makes Alice always accepts. Note we only need to distinguish the following two cases:
	
	\begin{itemize}
		\item For all $r \in [2^R]$, $\langle \vec{U}(x,r),\vec{V}(y,r) \rangle = W_{r}$.
		
		\item For at most half of $r \in [2^R]$, $\langle \vec{U}(x,r),\vec{V}(y,r) \rangle = W_{r}$.
	\end{itemize}

	Recall that $\vec{U}(x,r)$ and $\vec{V}(y,r)$ are vectors of $T$ degree $1$ polynomials from $\mathbb{Z}[x]$, with coefficients in $\{0,1, \dotsc,q-1\}$, and $W_{r}$ is a degree $2$ polynomial in $\mathbb{Z}[x]$, with coefficients in $\{0,1,\dotsc,O(q^{3})\}$. For a polynomial $P(x)$ in $\mathbb{Z}[x]$ and an integer $t$, let $[t]P(x)$ denote the coefficient of $x^t$ in $P(x)$. Then we can see $\langle \vec{U}(x,r),\vec{V}(y,r) \rangle = W_r$ is equivalent to the condition: for all $0 \le t \le 2$,
	\begin{equation}\label{eq:condition-W-r}
	\sum_{i=0}^{t} \sum_{k=1}^{T} [i] \vec{U}(x,r)_k \cdot [t-i] \vec{V}(y,r)_k = [t] W_r.
	\end{equation}
	
	Note that the left side of Equation~\eqref{eq:condition-W-r} is an inner product between two vectors from $\{0,1,\dotsc,q-1\}^{3 T}$. By Lemma~\ref{lm:encoding-trick}, we can construct three Boolean vectors $u_0,u_1,u_2 \in \{0,1\}^{O(q^6)}$ from $\vec{U}(x,r)$ and also $v_0,v_1,v_2 \in \{0,1\}^{O(q^6)}$ from $\vec{V}(y,r)$ and an integer $M$ (which only depends on $T$), such that:
	\begin{itemize}
		\item If Equation~\eqref{eq:condition-W-r} holds for all $t$, then $\sum_{i=0}^{2} \langle u_i, v_i \rangle = M$.
		\item Otherwise, $\sum_{i=0}^{2} \langle u_i, v_i \rangle < M$.
	\end{itemize}

	Now, we concatenate all these $u_0,u_1,u_2$ for all possibles $r$'s to form a single vector $u_x$, and construct $v_y$ similarly. We have:
	
	\begin{itemize}
		\item If for all $r \in [2^R]$, $\langle \vec{U}(x,r),\vec{V}(y,r) \rangle = W_{r}$, then $\langle u_x, v_y \rangle \ge 2^{R} \cdot M$.
		
		\item If for at most half of $r \in [2^R]$, $\langle \vec{U}(x,r),\vec{V}(y,r) \rangle = W_{r}$, then $\langle u_x, v_y \rangle \le 2^{R} \cdot (M - 1/2)$.
	\end{itemize}

	Now, let $A_\psi$ and $B_\psi$ be the collections of $u_x$ and $v_y$ with the proof $\psi$ respectively. Then we want to distinguish between the following two cases:
	
	\begin{itemize}
		\item There is a $\psi$ such that $\MAX(A_\psi,B_\psi) \ge 2^{R} \cdot M$.
		\item For all $\psi$, $\MAX(A_\psi,B_\psi) \le 2^{R} \cdot (M - 1/2)$.
	\end{itemize}

	Note that vectors in $A_\psi$ and $B_\psi$ are of dimension $ d = O(q^6 \cdot 2^{R})$, so the above can be solved by $2^{O(\eps \log n \log(c/\eps))} = n^{O(\eps \log (c/\eps))}$ calls to $\Omega(1/q^6) \cdot d = \Omega((\eps/c)^6) \cdot d$ additive approximation to $\MaxIP_{n,d}$, which completes the proof.
\end{proof}

Now we are ready to prove Theorem~\ref{theo:tighter-reduction-1}.

\begin{theo}\label{theo:tighter-reduction-1}
	The following are equivalent:
	\begin{enumerate}
		\item An $\eps \cdot d$ additive approximation to $\MaxIP_{n,d}$ is computable in $n^{2 - \eps^{o(1)}}$ time.
		\item $\MaxIP_{n,c\log n}$ is solvable in $n^{2 - 1/c^{o(1)}}$ time.
		\item $\ExactIP_{n,c\log n}$ is solvable in $n^{2 - 1/c^{o(1)}}$ time.
	\end{enumerate}
\end{theo}
\begin{proof}
	We only need to show that Item (1) implies Item (3). By Lemma~\ref{lm:ExactIP-to-additive-MaxIP}, there are constants $c_1,c_2$ such that for any constant $\eps_1 > 0$, every $\ExactIP_{n, c\log n}$ instance can be reduced to $n^{c_1 \eps_1 \log (c/\eps_1)}$ instances of $c_2 \cdot (\eps_1/c)^6 \cdot d$ additive approximations to $\MaxIP_{n,d}$ for $d = n^{o(1)}$.
	
	Suppose Item (1) holds, we set $\eps_1 = 1 / c$, then $\ExactIP_{n, c\log n}$ can be solved in
	\[
	n^{c_1\log (c^2) / c + 2 - (c_2 \cdot c^{-12})^{o(1)}} = n^{2 - 1 / c^{o(1)}}
	\]
	time, which completes the proof.
\end{proof}
	
\subsection{Tighter Connection Between Additive Approximation to $\MaxIP$ and Some Geometric Problems}\label{sec:additive-Max-IP-other-geometry}

Now we are ready to establish a similar connection between additive approximation to $\MaxIP$ and some geometric problems.

\begin{theo}\label{theo:tighter-reduction-2}
	The following are equivalent:
	\begin{enumerate}
		\item An $\eps \cdot d$ additive approximation to $\MaxIP_{n,d}$ is computable in $n^{2 - \eps^{o(1)}}$ time.
		\item An $\eps \cdot d$ additive approximation to $\MinIP_{n,d}$ is computable in $n^{2 - \eps^{o(1)}}$ time.
		\item A $(1+\eps)$ approximation to $\BCPp[p]$ is computable in $n^{2 - \eps^{o(1)}}$ time (for a constant $p \in [1,2]$).
		\item A $(1+\eps)$ approximation to $\FPp$ is computable in $n^{2 - \eps^{o(1)}}$ time (for a constant $p \in [1,2]$).
	\end{enumerate}
\end{theo}

One direction is simple, and already implicit in previous work.

\begin{lemma}[Theorem 4.1~\cite{Rub18BetterMA}]\label{lm:prev-work-tt}
	For any $p \in [1,2]$, if $\BCPp[p]$ or $\FPp[p]$ can be approximated in $n^{2 - \eps^{o(1)}}$ time, then there is an algorithm computing $\eps \cdot d$ additive approximation to $\MaxIP$ in $n^{2 - \eps^{o(1)}}$ time.
\end{lemma}

So it suffices to prove the other direction, we are going to apply Theorem~\ref{theo:LSH-to-MaxIP-general}.
\begin{proofof}{Theorem~\ref{theo:tighter-reduction-2}}
	The equivalence between Item (1) and (2) follows directly from Lemma~\ref{lm:MinIP-to-MaxIP}. By Lemma~\ref{lm:prev-work-tt}, Item (3) and (4) both imply Item (1). So it suffices to show Item (1) implies Item (3) and Item (4).
	
	We only consider $\BCPp[p]$ here; the case for $\FPp[p]$ are symmetric. Note that by a binary search (which incurs a negligible factor in the running time), we only need to consider the decision version, in which we are given a real $R$, and want to distinguish the two cases: (1) $\min_{(a,b) \in A \times B} \|a-b\|_p \le R$; (2) $\min_{(a,b) \in A \times B} \|a-b\|_p \ge (1+\eps) \cdot R$. 
	
	By Theorem~\ref{theo:LSH-to-MaxIP-general} and Lemma~\ref{lm:LSH-ell-p}, this decision problem can be reduced to computing an $\Omega(\eps \cdot d)$ approximation to $\MaxIP_{n,O(\eps^{-2} \log n)}$, which by assumption can be solved in $n^{2 - \eps^{o(1)}}$ time.
\end{proofof}

Finally, Theorem~\ref{theo:tighter-reduction} is a simple corollary of Theorem~\ref{theo:tighter-reduction-1} and Theorem~\ref{theo:tighter-reduction-2}.

%% file: data-structure.tex
\section{Equivalence in the Data Structure Setting}\label{sec:data-structure}

In this section, we generalize our equivalence results to the data structure setting.

We first introduce the data structure versions of $\OV$ and $\MaxIP$, which are used as intermediate problems for the reductions.

\begin{itemize}
	\item \textbf{Online $\OV$:} Preprocess a database $\mathcal{D}$ of $n$ points in $\{0,1\}^d$ such that, for all query of the form $q \in \{0,1\}^{d}$, either report a point $x \in \mathcal{D}$ which is orthogonal to $q$ or report that no $x$ exists.
	
	\item \textbf{Online $\MaxIP$:} Preprocess a database $\mathcal{D}$ of $n$ points in $\{0,1\}^d$ such that, for all query of the form $q \in \{0,1\}^{d}$, find a point $x \in \mathcal{D}$ maximizing $\langle x, q \rangle$.
\end{itemize}

\begin{theo}\label{theo:eq-OnlineOV-OnlineMaxIP-NNS}
	The following are equivalent:		
	\begin{itemize}
		\item There is a $\delta > 0$ such that for all constant $c$, there is a data structure for Online $\OV$ with $d = c \log n$ uses $\poly(n)$ space and allows $n^{1-\delta}$ query time.
		
		\item There is a $\delta > 0$ such that for all constant $c$, there is a data structure for Online $\MaxIP$ with $d = c \log n$ uses $\poly(n)$ space and allows $n^{1-\delta}$ query time.
		
		\item There is a $\delta > 0$ such that for all $\epsilon > 0$, there is a data structure for approximate NNS in $\ell_p$ with approximation ratio $(1+\epsilon)$ uses $\poly(n)$ space and allows $n^{1-\delta}$ query time for a constant $p \in [1,2]$.
	\end{itemize}		
\end{theo}

Note that by~\cite{abboud2015more}, Online $\OV$ is equivalent to Partial Match, so the above theorem implies Theorem~\ref{theo:eq-partial-match-NNS}.

We also need the following two important observations from the proof of Lemma~\ref{lm:ExactIP-to-OV} and Lemma~\ref{lm:ExactIP-to-additive-MaxIP}.

\begin{lemma}[Implicit in Lemma~\ref{lm:ExactIP-to-OV}]\label{lm:specific-construction-1}
	Let $n$ be an integer, $c$ be a constant, $\epsilon > 0$ and $0 \le k \le c\log n$. There are two families of functions $f_1,f_2,\dotsc,f_{m}$ and $g_1,g_2,\dotsc,g_{m}$ from $\{0,1\}^{c\log n}$ to $\{0,1\}^{2^{O(c/\epsilon)} \log n}$ where $m = n^{O(\eps \log (c/\epsilon))}$, such that for all $x,y \in \{0,1\}^{c \log n}$, 
	$\langle x, y \rangle = k$ if and only if there is an $i \in [m]$ such that $ \langle f_i(x), g_i(y) \rangle = 0$. Moreover, functions $f_i$'s and $g_i$'s can be evaluated in $\polylog(n)$ time.
\end{lemma}

\begin{lemma}[Implicit in Lemma~\ref{lm:ExactIP-to-additive-MaxIP} and \ref{lm:prev-work-tt}]\label{lm:specific-construction-2}	
	Let $p \in [1,2]$, $n$ be an integer, $c$ be a constant, $\epsilon > 0$ and $0 \le k \le c\log n$. There are two families of functions $f_1,f_2,\dotsc,f_{m}$ and $g_1,g_2,\dotsc,g_{m}$ from $\{0,1\}^{c\log n}$ to $\R^{n^{o(1)}}$ where $m = n^{O(\eps \log (c/\epsilon))}$, such that for all $x,y \in \{0,1\}^{c \log n}$, 	
	\begin{itemize}
		\item If $\langle x, y \rangle = k$, then there is an $i \in [m]$ such that $\| f_i(x) - g_i(y) \|_p \le 1 - \Omega((\epsilon/c)^6)$.
		\item Otherwise, for all $i \in [m]$, $\| f_i(x) - g_i(y) \|_p \ge 1$.
	\end{itemize}
	 Moreover, functions $f_i$'s and $g_i$'s can be evaluated in $n^{o(1)}$ time.
\end{lemma}

\begin{proofof}{Theorem~\ref{theo:eq-OnlineOV-OnlineMaxIP-NNS}} 
		
	In the below we first show the equivalence between Online $\OV$ and Online $\MaxIP$, the equivalence between Online $\MaxIP$ and NNS is proved similarly, so we only sketch the main ideas.
	
	\paragraph*{Online $\OV$ $\Leftrightarrow$ Online $\MaxIP$.} The reduction from Online $\OV$ to Online $\MaxIP$ is trivial. For the other direction, suppose there is a $\delta > 0$ such that for all constant $c$, there is an algorithm for Online $\OV$ with $d = c \log n$ such that it uses $\poly(n)$ space and allows $n^{1-\delta}$ query time.
	
	Let $d = c \log n$ for a constant $c$, and $c_1$ be the constant hiding in the big-$O$ of $m=2^{O(\epsilon\log(c/\epsilon))}$ in Lemma~\ref{lm:specific-construction-1}. Suppose we are given a set $\mathcal{D}$ of $n$ points from $\{0,1\}^d$.
	
	We set $\epsilon$ such that $c_1 \cdot \epsilon \log (c/\epsilon) = \delta/2$ and apply Lemma~\ref{lm:specific-construction-1}. Now, for each $0 \le k \le d$, we build $n^{c_1 \cdot \epsilon \log (c/\epsilon)} = n^{\delta/2}$ data structures for Online $\OV$, the $i$-th data structure consists of the $f_i(x)$'s for all $x \in \mathcal{D}$. Note that the $f_i(x)$'s have length $2^{O(c/\epsilon)} \cdot \log n$, which is still $O(\log n)$ as $\epsilon$ is a constant.
	
	For each query $q \in \{0,1\}^{d}$, note that there is an $x \in \mathcal{D}$ such that $\langle x, q \rangle = k$ if and only if there is an $i$ such that the $i$-th Online $\OV$ structure contains an orthogonal point to $g_i(q)$. Therefore, by enumerating $k$ from $d$ down to $0$, $i$ from $\left[n^{\delta/2}\right]$, and making corresponding queries to the Online $\OV$ data structures, one can answer queries for Online $\MaxIP$ in $n^{1 - \delta/2} \cdot d$ time.
	
	\paragraph*{Online $\MaxIP$ $\Leftrightarrow$ Approximate NNS (Sketch).} Using Lemma~\ref{lm:specific-construction-2}, the reduction from Online $\MaxIP$ to Approximate NNS can be proved similarly as from Online $\MaxIP$ to Online $\OV$. 
	
	For the direction from approximate NNS to Online $\MaxIP$: suppose the approximation ratio is $(1+\epsilon)$. It suffices, for all $R$ of the form $(1+\epsilon/3)^{k}$ for an integer $k$, to construct a data structure which finds a point with distance smaller than $R \cdot (1+\epsilon/3)$ if the minimum distance is smaller than $R$, and reports a failure if the minimum distance is greater than $R \cdot (1+\epsilon/3)$ (its behavior can be arbitrary if neither case holds). Using the reduction implicit in proof of Theorem~\ref{theo:LSH-to-MaxIP-general}, this can be reduced to Online $\MaxIP$ with $d = O(\log n)$. 
	
\end{proofof}

%% file: Apx-Min-IP.tex
\section{Algorithms for $\MAMinIP$ and $\MAMaxIP$}
	
	In this section we give fast algorithms for $\MAMinIP$ and $\MAMaxIP$. Our algorithms make use of the polynomial method~\cite{abboud2015more}. 	For simplicity of exposition, we set the approximation factors in $\MAMinIP$ and $\MAMaxIP$ to be $2$, but our algorithms can be extended to work for any constant approximation factor $\kappa > 1$ easily.
	
	
	\subsection{Low Degree Probabilistic Polynomial Implies Fast Algorithms} 
	Abboud, Williams, and Yu~\cite{abboud2015more},  show that for a Boolean vector problem, a ``sparse'' probabilistic polynomial for the problem implies a fast algorithm. To state their result formally, we first introduce some notations.
	
	For our purposes, we will think of a \emph{probabilistic polynomial} $\mathcal{P}$ as a distribution over $\mathbb{F}_2$-polynomials (polynomials over the field $\mathbb{F}_2$), and the \emph{degree} of a probabilistic polynomial is the maximum degree of all polynomials in its support. For a function $f : D \to \{0,1\}$, we say $\mathcal{P}$ is an $\eps$-error probabilistic polynomial for $f$, if for every $x \in D$, $\Pr_{P \sim \mathcal{P}} [P(x) \ne f(x)] \le \eps$.
		
	Let us abstract out a key result from~\cite{abboud2015more}, for our use here:
	
	\begin{theo}[\cite{abboud2015more}]\label{theo:AWY15}
		Let $c$ be an integer and $d = c \log n$, let $f : D \to \{0,1\}$ with $D \subseteq \{0,1\}^{d} \times \{0,1\}^{d}$ be a function. 
        Suppose that:
		
		\begin{itemize}
			\item For any $\eps > 0$, there is an $\eps$-error probabilistic polynomial $\mathcal{P}$ for $f$ with degree $t = O(\log \eps^{-1})$.
			\item A sample from $\mathcal{P}$ can be generated in $\poly(\binom{d}{ \le t})$ time.\footnote{$\binom{n}{\le m}$ denotes $\sum_{i=0}^{m} \binom{n}{i}$.}
		\end{itemize}
		
		Then there is an algorithm $\alg$ such that:
		
		\begin{itemize}
			\item Given two sets $A$ and $B$ of $n$ vectors from $\{0,1\}^{d}$, $\alg$ runs in $n^{2 - 1/O(\log c)}$ time.
			
			\item If for every $(a,b) \in A \times B$, $f(a,b) = 0$, then $\alg$ outputs $0$ with probability at least $1 - 1/n$.
			
			\item If there is an $(a,b) \in A \times B$, $f(a,b) = 1$, then $\alg$ outputs $1$ with probability at least $1 - 1/n$.\footnote{If neither of the above two cases hold, the algorithm can output anything.}
		\end{itemize}
		
	\end{theo}
	
	\newcommand{\GapMinIp}{\textsf{Gap-Min-IP}}
	\newcommand{\GapMaxIp}{\textsf{Gap-Max-IP}}
	
	\newcommand{\fgap}{f^{\textsf{gap}}}
	
	\subsection{$n^{2-1/O(\log c)}$ Time Algorithms for $\MAMinIP$ and $\MAMaxIP$}
	
	In order to apply the theorem above, we need to switch from $\MAMinIP$ and $\MAMaxIP$ to their closely related decision problems $\GapMinIp$ and $\GapMaxIp$.
	
	\begin{defi}
		For $n,d \in \mathbb{N}$, we define the problems:
		\begin{itemize}
			\item  $\GapMinIp_{n,d}$ : \emph{Given sets $A$ and $B$ of $n$ vectors from $\{0,1\}^d$ and an integer $\tau$, and the promise that either $\MIN(A,B) \le \tau$ or $\MIN(A,B) \ge 2\tau$, the task is to decide which.}
			\item  $\GapMaxIp_{n,d}$ : \emph{Given sets $A$ and $B$ of $n$ vectors from $\{0,1\}^d$ and an integer $\tau$, and the promise that either $\MAX(A,B) \le \tau$ or $\MAX(A,B) \ge 2\tau$, the task is to decide which.}
		\end{itemize}
		
		Moreover, for two vectors $x,y \in \{0,1\}^{d}$ and an integer $\tau$, we define the corresponding gap-deciding function:
		
		\[
		\fgap_{d,\tau}(x,y) = \begin{cases}
		1 &\qquad \text{$\langle x, y \rangle \ge 2\tau$,}\\
		0 &\qquad \text{$\langle x, y \rangle \le \tau$,}\\
		\text{undefined} &\qquad \text{otherwise.}
		\end{cases}
		\]
		
		When $d$ and $\tau$ are clear from the context, we omit them for simplicity.
	\end{defi}

	\begin{rem}\label{rm:fgap-formulation}
		$\GapMaxIp_{n,d}$ ($\GapMinIp_{n,d}$) is equivalent to determine whether there is an $(a,b) \in A \times B$ such that $\fgap(a,b) = 1$ ($\fgap(a,b) = 0$) or for all $(a,b) \in A \times B$ we have $\fgap(a,b) = 0$ ($\fgap(a,b) = 1$). 
	\end{rem}
	
	\newcommand{\micro}{\textsf{micro}}
	
	The following lemma is the key technical ingredient of this section.
	
	\begin{lemma}\label{lm:polynomial-for-fgap}
		For all $d,\tau \in \mathbb{N}$ and $\eps \in (0,1/10)$, there is a $t = O(\log \eps^{-1})$-error probabilistic polynomial $\mathcal{P}$ for $\fgap_{d,\tau}$. Moreover, a sample from $\mathcal{P}$ can be generated in $\poly(\binom{d}{\le t})$ time.
	\end{lemma}
	
	Before proving Lemma~\ref{lm:polynomial-for-fgap}, we first show it implies Theorem~\ref{theo:fast-algo-MAMinIP} (restated below) together with Theorem~\ref{theo:AWY15}.
	
	\begin{reminder}{Theorem~\ref{theo:fast-algo-MAMinIP}.}
		There are $n^{2 - 1/O(\log c)}$ time randomized algorithms for $\MAMinIP_{n,c \log n}$ and $\MAMaxIP_{n,c \log n}$.
	\end{reminder}
	
	\begin{proofof}{Theorem~\ref{theo:fast-algo-MAMinIP}}
		We only consider $\MAMinIP$ here; the case for $\MAMaxIP$ is symmetric.
		By Lemma~\ref{lm:polynomial-for-fgap}, Theorem~\ref{theo:AWY15}, and Remark~\ref{rm:fgap-formulation}, there is a randomized algorithm $\alg$ for $\GapMinIp_{n,c\log n}$ in $n^{2 - 1/O(\log c)}$ time.
		
		\newcommand{\taumin}{\tau_{\textsf{min}}}
		
		Now we turn $\alg$ into an algorithm for $\MinIP$. We say $\alg$ outputs $1$ if it decides $\MIN(A,B) \le \tau$, and $0$ otherwise. We enumerate $\tau$ from $0$ to $d$, and let $\taumin$ be the smallest $\tau$ such that $\alg$ outputs $1$. Note that such $\tau$ exists, as $\alg$ must output $1$ when $\tau = d$.
		
		With probability at least $1 - (d+1)/ n \ge 2/3$, $\alg$ operates correctly on all enumerated $\tau$'s. We condition on that event in the following. 
		Since $\alg$ outputs $1$ with $\taumin$, we have $\MIN(A,B) < 2\taumin$ (otherwise it must output $0$). Similarly, as $\alg$ outputs $0$ with $\taumin - 1$ ($\taumin$ is the smallest), we have $\MIN(A,B) > \taumin -1$ (otherwise it must output $1$). Therefore, we can see $2 \taumin \in [\MIN(A,B),2 \cdot \MIN(A,B)]$ with probability at least $2/3$, and we obtain an $n^{1 - 1/O(\log c)}$ algorithm for $\MAMinIP$.
	\end{proofof} 

	
	Finally, we devote the rest of this section to the proof of Lemma~\ref{lm:polynomial-for-fgap}.
	
	\begin{proofof}{Lemma~\ref{lm:polynomial-for-fgap}}
		In the following, we assume $\tau \le d/2$ as the function becomes trivial otherwise.
		
		We begin by introducing some notation. Let $T$ be a sequence of integers from $[d]$, we use $x_{|T} \in \{0,1\}^{|T|}$ to denote the projection of $x$ on $T$, such that $(x_{|T})_i := x_{T_i}$ for $i \in [|T|]$. We also use the Iverson bracket notation: for a predicate $P$, $\left[ P \right]$ takes value $1$ when $P$ is true, and $0$ otherwise.
		
		\highlight{Construction of ``Micro'' Probabilistic Polynomial $\mathcal{P}_\micro$.} The first step is to construct a probabilistic polynomial $\mathcal{P}_\micro$ of degree $1$, such that for $x,y \in \{0,1\}^{d}$:
		
		\begin{itemize}
			\item If $\langle x, y \rangle \ge 2 \tau$: $\Pr_{P \sim \mathcal{P}_\micro} [ P(x,y) = 1 ] \ge c_1$ for a universal constant $c_1$.
			\item If $\langle x, y \rangle \le \tau$: $\Pr_{P \sim \mathcal{P}_\micro} [ P(x,y) = 1 ] \le c_2$ for a universal constant $c_2$.
			\item $c_1 > c_2$.
		\end{itemize}
		
		Let $k = \frac{d}{\tau}$. By our assumption, we have $k \ge 2$. Now a sample from $\mathcal{P}_\micro$ is generated as follows:	
		\begin{itemize}
			\item We pick a sequence $T$ of $k$ uniform random numbers from $[d]$ and a uniform random vector $z \in \{0,1\}^{k}$.
			
			\item We set $P(x,y) := \sum_{i=1}^{k} z_i \cdot (x_{|T})_i \cdot (y_{|T})_i$ (which is an $\mathbb{F}_2$ polynomial).
		\end{itemize}
		
		First, we make the following observations:
		
		\begin{itemize}
			\item If $\langle x, y \rangle \ge 2 \tau$:
			\begin{align*}
			\Pr_{T}[ \langle x_{|T}, y_{|T} \rangle > 0] &\ge 1 - \left(1 - \frac{2\tau}{d} \right)^{k} \\
											 &= 1 - \left(1 - \frac{2}{k} \right)^k \ge 1 - e^{-2} > 0.86.
			\end{align*}
			\item If $\langle x, y \rangle \le \tau$:
			\begin{align*}
			\Pr_{T}[ \langle x_{|T}, y_{|T} \rangle > 0] &\le 1 - \left(1 - \frac{\tau}{d} \right)^{k} \\
											 &= 1 - \left(1 - \frac{1}{k} \right)^k \le 1 - \frac{1}{4} = 0.75.
			\end{align*}
		\end{itemize}
		
		Note that when $\langle x_{|T}, y_{|T} \rangle = 0$, $P(x,y)$ is always $0$, and when $\langle x_{|T}, y_{|T} \rangle > 0$, $P(x,y) = 1$ with probability $1/2$. Therefore, we have:

		\begin{itemize}
			\item If $\langle x, y \rangle \ge 2 \tau$:
			\[
			\Pr_{P \sim \mathcal{P}_\micro} [ P(x,y) = 1 ] \ge (0.86) / 2 = 0.43.
			\]
			\item If $\langle x, y \rangle \le \tau$:
			\[
			\Pr_{P \sim \mathcal{P}_\micro} [ P(x,y) = 1 ] \le (0.75) / 2 = 0.375.
			\]
		\end{itemize}
		
		The above completes our construction of the ``micro'' probabilistic polynomial $P_\micro$.
		
		\highlight{Construction of the Probabilistic Polynomial $\mathcal{P}_{\textsf{final}}$.} Now, let $m = c_1 \cdot \log \eps^{-1}$ for a sufficiently large constant $c_1$. And let $P_{1},P_{2},\dotsc,P_{m}$ be $m$ i.i.d. samples from $\mathcal{P}_{\micro}$.
		By a simple Chernoff bound, we have:
		
		\begin{itemize}
			\item If $\langle x, y \rangle \ge 2 \tau$:
			\[
			\Pr_{P_1,P_2,\dotsc,P_m \sim \mathcal{P}_\micro} \left[ \sum_{i=1}^{m} \left[P_i(x,y) = 1\right] > 0.4 \cdot m \right] \ge 1 - \eps.
			\]
			\item If $\langle x, y \rangle \le \tau$:
			\[
			\Pr_{P_1,P_2,\dotsc,P_m \sim \mathcal{P}_\micro} \left[ \sum_{i=1}^{m} \left[P_i(x,y) = 1\right] < 0.4 \cdot m \right] \ge 1 - \eps.
			\]
		\end{itemize}
		
		\newcommand{\Pfinal}{P_{\textsf{final}}}
		
		Finally, we set
		\[
		\Pfinal(x,y) := \sum_{ S \subseteq [m], |S| > 0.4 m } \prod_{i \in S} P_i(x,y) \cdot \prod_{i \notin S} (1 - P_i(x,y)).
		\]
		
		Clearly, $\Pfinal(x,y) = 1$ if and only if $\sum_{i=1}^{m} \left[P_i(x,y) = 1\right] > 0.4 \cdot m$, and therefore its distribution $\mathcal{P}_{\textsf{final}}$ is the $\eps$-error probabilistic polynomial we want. And it is easy to see a polynomial from $\mathcal{P}_{\textsf{final}}$ can be sampled in the stated time.		
		
		\end{proofof}
	 
	 \subsection{A Fast Algorithm for Approximating ``Almost Satisfiable'' $\MAXSAT$ Instances}

	 Finally, we give an application of the algorithm for $\MAMinIP$ by proving Theorem~\ref{theo:fast-algo-MAXSAT-eps} (restated below).
	 
	 \begin{reminder}{Theorem~\ref{theo:fast-algo-MAXSAT-eps}}
	 	Let $\varphi$ be a $\MAXSAT$ instance on $n$ variables with $m$ clauses, and $\eps = 1 - \sat(\varphi)$. There is a $2^{n (1 - 1 / O(\log \eps^{-1}) )}$ time algorithm to find an assignment $x$ satisfying at least $(1 - 2 \eps) \cdot m$ clauses.
	 \end{reminder}
	 \begin{proof}
	 	We use the reduction from $\textsf{CNF-SAT}$ to $\OV$, from~\cite{Wil05}. For simplicity, suppose $2$ divides $n$. For an assignment $x$ to $\varphi$, we use $\val(\varphi,x)$ to denote the number of satisfied clauses of $\varphi$ by $x$, divided by $m$.
	 	
	 	First, we do a ``sparsification'' step: we pick $M = c_1 \cdot \eps^{-2} \cdot n$ clauses from $\varphi$ at uniformly random. Let $\psi$ be the $\MAXSAT$ instance with these randomly chosen clauses.
	 	
	 	By a standard Chernoff bound, with a sufficiently large universal constant $c_1$, for every assignment $x \in \{0,1\}^n$, we have
	 	\[
	 	\Pr\left[ | \val(\varphi,x) - \val(\psi,x) | \le \eps/3 \right] \le 1/2^{2n}.
	 	\]
	 	
	 	Therefore, by a union bound, with probability at least $1 - 1/2^{n}$, for all $x \in \{0,1\}^n$ we have $| \val(\varphi,x) - \val(\psi,x) | \le \eps/3$, and it follows that $|\sat(\varphi) - \sat(\psi)| \le \eps / 3$.
	 	So it suffices to consider $\psi$ now.
	 	
	 	Next, we split these $n$ variables into two groups 
	 	\[ 
	 	x_L := \{x_1,\dotsc,x_{n/2} \} \quad\text{ and }\quad x_R := \{ x_{n/2+1},\dotsc,x_{n} \}.
	 	\]		
	 	Let $\cla_1,\cla_2,\dotsc,\cla_M$ be all clauses in $\psi$. For each $a \in \{0,1\}^{n/2}$, interpreted as an assignment to variables in $x_L$, we construct a vector $u_a \in \{0,1\}^{M}$, such that $(u_a)_i = 1$ iff $\cla_i$ is not satisfied when setting variables in $x_L$ according to $a$. Similarly, for each $b \in \{0,1\}^{n/2}$, we interpret it as an assignment to variables in $x_R$, and construct a vector $v_b \in \{0,1\}^{M}$ in the same way.
	 	
	 	Next, for $a,b \in \{0,1\}^{n/2}$, $\langle (u_a)_i, (v_b)_i \rangle = 1$ if and only if $\cla_i$ is not satisfied by the joint assignment $(a,b)$. Therefore, $\langle u_a, v_b \rangle$ is the number of clauses that are not satisfied by the joint assignment $(a,b)$.
	 	
	 	Let $A$ be the set of all $u_a$'s for $a \in \{0,1\}^{n/2}$, and $B$ be the set of all $v_b$'s for $b \in \{0,1\}^{n/2}$. By Theorem~\ref{theo:fast-algo-MAMinIP}, there is an algorithm which finds a $(u_a,v_b) \in A \times B$ such that $ \langle u_a, v_b \rangle \in [\MIN(A,B),1.1 \cdot \MIN(A,B)]$.
	 	
	 	From the definition, we have $\MIN(A,B) := (1 - \sat(\psi)) \cdot M \le \frac{4}{3} \cdot \eps \cdot M$ (recall that $\eps = 1 - \sat(\varphi)$). Therefore, we have $\langle u_a, v_b \rangle \le 1.1 \cdot \frac{4}{3} \cdot \eps \cdot M \le 1.5 \cdot \eps \cdot M $.
	 	
	 	Let $x$ be the joint assignment $(a,b)$. We have $\val(\psi,x) \ge (1 - 1.5 \eps)$. Since $| \val(\psi,x) - \val(\varphi,x)| \le \eps/3$, $\val(\varphi,x) \ge (1 - 2\eps)$, which means $x$ is a valid answer.
	 	
	 	Finally, as $M = O(\eps^{-2}) \cdot n$, the algorithm runs in $\left( 2^{n/2} \right)^{2 - 1/O(\log \eps^{-1})} = 2^{n \cdot (1-1/O(\log \eps^{-1}))}$ time, which completes the proof.
	 \end{proof}

%% file: app.tex
\section{Missing Proofs}
	Here we give some missing proofs in the paper. 
	
	\begin{proofof}{Lemma~\ref{lm:bound}}
		Let $X = \sum_{i=1}^{cm} X_i$ and $\mu = \Ex[X] = \eps \cdot cm$. Set $\delta = \eps^{-1} / 3$. By the multiplicative Chernoff bound, we have
		\[
		\Pr[X > (1 + \delta) \cdot \mu] < \left( \frac{e^\delta}{(1+\delta)^{1+\delta}} \right)^{\mu}.
		\]
		
		Note that $(1 +\delta) \cdot \mu = (1 + \eps^{-1}/3) \cdot \eps \cdot cm < \frac{1}{2} \cdot c m$. Also, we have
		\begin{align*}
		\left( \frac{e^\delta}{(1+\delta)^{1+\delta}} \right)^{\mu} &= e^{-\mu \cdot \left[ (1+\delta) \ln(1+\delta) - \delta \right]}\\
		&\le e^{ - \mu \cdot \left[ \delta \ln \delta - \delta \right]}\\
		&\le e^{ - \eps \cdot cm \cdot \left[ \eps^{-1} / 3 \ln(\eps^{-1}/3) \right] / 2}\\
		&\le e^{ -\eps \cdot cm \cdot \left[ \eps^{-1} \ln(\eps^{-1}) \right] / 12}\\
		&\le e^{ - \ln \eps^{-1} \cdot cm / 12} = \eps^{-cm/12}.
		\end{align*}
		
		Therefore, we can set $c = 12$, and the proof is completed.
	\end{proofof}
	
	\begin{proofof}{Lemma~\ref{lm:MinIP-to-MaxIP}}
		We define two functions $\varphi_x,\varphi_y : \{0,1\} \to \{0,1\}^2$ such that:
		\[
		\varphi_x(0) := (1,0),\quad \varphi_x(1) := (0,1),\quad \varphi_y(0) := (1,1),\quad \varphi_y(1) := (1,0).
		\]
		
		It is easy to check that for $a,b \in \{0,1\}$, $a \cdot b = 1 -  \langle \varphi_x(a), \varphi_y(b) \rangle$. Then, for $x,y \in \{0,1\}^{d}$, we define $\psirevx(x) \in \{0,1\}^{2d}$ as the concatenation of $\varphi_x(x_i)$ for each $i \in [d]$, and similarly $\psirevy(y) \in \{0,1\}^{2d}$ as the concatenation of $\varphi_y(y_i)$ for each $i \in [d]$.
		
		Then we can see $ \langle \psirevx(x), \psirevy(y) \rangle = \sum_{i=1}^{d} \langle \varphi_x(x_i), \varphi_y(y_i) \rangle =  d -  \langle x, y \rangle$.
	\end{proofof}

	\begin{proofof}{Lemma~\ref{lm:ExactIP-to-MinIP}}		
		We remark the reduction here is essentially the same as the trick used in~\cite{Wil18}. For a vector $v \in \{0,1\}^*$, we use $v^{\otimes k}$ to denote the concatenation of $k$ copies of $v$. 
		
		Consider the following polynomial $P(x,y) := ( \langle x, y \rangle - m)^2$, we have
		\[
		P(x,y) = \langle x, y \rangle^2 - 2  m \langle x, y \rangle + m^2 = \langle x, y \rangle^2 + 2  m (d -  \langle x, y \rangle ) + m^2 - 2dm.
		\]
		
		For convenience, for a vector $z \in \{0,1\}^{d^2}$, we use $z_{(i,j)}$ to denote the $(i-1) \cdot d + j$-th coordinate of $z$. For $x,y \in \{0,1\}^{d}$, we construct $\WT{x},\WT{y} \in \{0,1\}^{d^2}$ such that $\WT{x}_{(i,j)} = x_{i} \cdot x_{j}$ and $\WT{y}_{(i,j)} = y_i \cdot y_j$. Then we can see 
		$$ 
		\langle \WT{x}, \WT{y} \rangle = \sum_{i=1}^{d}\sum_{j=1}^d x_i x_j \cdot y_i y_j = \left( \sum_{i=1}^{d} x_i y_i \right)^2 = \langle x, y \rangle^2.
		$$  Let $\psirevx$ and $\psirevy$ be the two functions from Lemma~\ref{lm:MinIP-to-MaxIP}. For $x,y \in \{0,1\}^d$, we define
		\[
		\varphi^x_{d,m}(x) := (\WT{x},\psirevx(x)^{\otimes (2m)}) 
		\qquad\text{ and }\qquad 
		\varphi^y_{d,m}(y) := (\WT{y},\psirevy(y)^{\otimes (2m)}).
		\]
		
		Then we have $ \langle \varphi^x_{d,m}(x), \varphi^y_{d,m}(y) \rangle = \langle x, y \rangle^2 + 2  m (d -  \langle x, y \rangle ) = P(x,y) + 2dm - m^2$. And we set $M_{d,m} = 2dm - m^2$.
		
		Now, if $ \langle x, y \rangle = m$, we have $P(x,y) = 0$, and therefore $ \langle \varphi^x_{d,m}(x), \varphi^y_{d,m}(y) \rangle = M_{d,m}$. Otherwise, $ \langle x, y \rangle \ne m$ and we have $P(x,y) > 0$, and hence $ \langle \varphi^x_{d,m}(x), \varphi^y_{d,m}(y) \rangle > M_{d,m}$. Note that $\varphi^x_{d,m}(x),\varphi^y_{d,m}(y) \in \{0,1\}^{d^2 + 4dm}$, we add $5d^2 - M_{d,m}$ dummy ones to the end of $\varphi^x_{d,m}(x)$ and $\varphi^y_{d,m}(y)$ and set $M_d = 5d^2$, which completes the proof.
	\end{proofof}

	\begin{proofof}{Lemma~\ref{lm:encoding-trick}}		
		We begin by the construction of two embeddings $\psi_x,\psi_y : \{0,1,\dotsc, r\} \to \{0,1\}^{r^2}$ such that for any $x,y \in \{0,1,\dotsc,r \}$, $ \langle \psi_x(x), \psi_y(y) \rangle = x \cdot y$.
		
		For convenience, in the following we use $z_{(i,j)}$ to denote the $(i-1) \cdot r + j$-th coordinate of $z$. Then we define $\psi_x(x)_{(i,j)}$ as $1$ when $i \le x$, and $0$ otherwise; similarly, we define $\psi_y(y)_{(i,j)}$ as $1$ when $j \le y$, and $0$ otherwise. We have
		
		\[
		 \langle \psi_x(x), \psi_y(y) \rangle = \sum_{i=1}^r\sum_{j=1}^r \psi_x(x)_{(i,j)} \cdot \psi_y(y)_{(i,j)} = \sum_{i=1}^{x}\sum_{j=1}^{y} 1 \cdot 1 =  \langle x, y \rangle.
		\]
		
		Slightly abusing notations, for $x,y \in \{0,1,\dotsc,r\}^{d}$, we define $\psi_x(x)$ and $\psi_y(y)$ as the concatenation of $\psi_x$ or $\psi_y$ applying on all coordinates of $x$ and $y$. Then we have $\psi_x(x),\psi_y(y) \in \{0,1\}^{d r^2}$, and $ \langle \psi_x(x), \psi_y(y) \rangle =  \langle x, y \rangle$. Then applying Lemma~\ref{lm:ExactIP-to-MinIP} and Lemma~\ref{lm:MinIP-to-MaxIP} completes the proof.
	\end{proofof}
		

%% file: moreapp.tex
\section{More Applications of the $\Sigma_2^\cc$ Reduction Framework}
\label{app:moreapp}

To demonstrate the potential power of our $\Sigma_2^\cc$ framework. In the following we discuss some of its applications other than establishing our equivalence class. The first one is a very simple reduction from Hopcroft's problem in \emph{constant dimensions} to $\OV$ in \emph{polylogarithmic dimensions}. And the second one is a reduction from $\textsf{3-SUM}$ to 3-$\OV$.

\subsection{Integer Inner Product and Hopcroft's Problem}

The Hopcroft's problem is defined as follows: you are given two sets $A,B$ of $n$ vectors from $\mathbb{Z}^{d}$, and want to determine whether there is an $(a,b) \in A \times B$ such that $\langle a, b \rangle = 0$. In other words, it is the same as $\OV$ except for now vectors consist of integer entries. 

We use $\Hopcroft_{n,d}$ to denote this problem in $d$ dimensions for simplicity, and assume the integers in $\Hopcroft_{n,d}$ belong to $[-n^c,n^c]$ for a constant $c$, which is the most interesting case. Now we formally state our reduction.

\begin{theo}\label{theo:Hopcroft-to-ov}
	Let $c,d$ be two constants, a $\Hopcroft_{n,d}$ instance $I$ with entries in $[-n^c,n^c]$ can be reduced to an $\OV_{n,O(\log n)^{d+1}}$ instance $J$ in $n^{1+o(1)}$ time, such that $I$ is a yes instance if and only if $J$ is a yes instance.
\end{theo}

An immediate corollary is that if moderate dimensional $\OV$ is in truly subquadratic time, then $\Hopcroft_{n,d}$ is also in truly subquadratic time for all constant $d$.

\newcommand{\ZIP}{\textsf{Z-IP}}

Let $d$ be an integer, we define $\ZIP_{d} : \mathbb{Z}^{d} \times \mathbb{Z}^{d} \to \{0,1\}$ as the function that checks whether two $d$ dimensional vectors in $\mathbb{Z}^{d}$ are orthogonal. Note that $\Hopcroft_{n,d}$ is equivalent to $\ZIP_{d}\SATPAIR_{n}$.

Theorem~\ref{theo:Hopcroft-to-ov} is just a direct corollary of the following fast $\Sigma_2$ communication protocol for $\ZIP_{d}$ (it is in fact a $\coNP$ communication protocol, as Merlin sends nothing) and Theorem~\ref{theo:Sigma2-to-OV}.

\begin{lemma}\label{lm:com-protocol}
	Let $c,d$ be two constants, there is a $\Sigma_2^\cc$ protocol for $\ZIP_{d}$ with entries in $[-n^{c},n^{c}]$, in which Merlin sends nothing, Megan sends $\log\log(n) + O(1)$ bits and Alice and Bob communicate $d \cdot \log\log(n) + O(d)$ bits.
\end{lemma}
\begin{proof}
	Let $x,y$ be two vectors from $[-n^c,n^c]^{d}$, we have $|\langle x, y \rangle| \le d \cdot n^{2c}$.
	
	Let $t$ be the smallest number such that the first $t$ primes $p_1,p_2,\dotsc,p_{t}$ satisfy $\prod_{i=1}^{t} p_i > d \cdot n^{2c}$. We first bound $t$ and the largest prime $p_t$. Clearly, $t \le \log (d \cdot n^{2c}) = O(\log n)$. Recall that $n\#$ denotes the product of all primes less than or equal to $n$ (the primordial function), and we have $n\# = e^{(1+o(1)) n}$. By the definition of $t$, it follows that $(p_t-1)\# = e^{(1+o(1)) \cdot (p_t-1)} \le d \cdot n^{2c} $ and $p_t = O(\log n)$.

	From our choice of $t$, we have $\langle x, y \rangle = 0$ if and only if $\langle x, y \rangle \equiv 0 \pmod{p_i}$ for all $i \in [t]$. So in the protocol, Merlin sends nothing. Megan sends an index $i \in [t]$, which takes $\log t = \log \log n + O(1)$ bits. After that, for each $j \in [d]$, Bob sends $y_j \bmod p_i$ to Alice, which takes $d \cdot \log p_i \le d \cdot \log \log n + O(d)$ bits, and Alice accepts if and only if $\langle x, y \rangle \equiv 0 \pmod{p_i}$. 
\end{proof}

\subsection{Sum-Check and 3-Sum}

\newcommand{\SATTRIPLE}{\textsf{-Satisfying-Triple}}

Next we discuss a reduction from $3$-SUM to $3$-$\OV$. $3$-$\OV$ is a generalized version of $\OV$, in which you are given three sets $A,B,C$, each of $n$ vectors from $\{0,1\}^{d}$, and want to determine whether there is an $(a,b,c) \in A \times B \times C$ such that $\sum_{i=1}^{d} a_i \cdot b_i \cdot c_i = 0$ (the generalized inner product of $a,b$, and $c$ is zero). We use $\thOV_{n,d}$ to denote the $\thOV$ problem with sets of $n$ vectors of $d$ dimensions.

\begin{theo}\label{theo:3-sum-to-3-ov}
	If $3$-$\OV$ is in truly-subquadratic time\footnote{This means there is an $\eps > 0$ such that for all constants $c$, $\thOV_{n,c\log n}$ can be solved in $n^{2-\epsilon}$ time.}, then so is $3$-$\textsf{SUM}$.
\end{theo}

We remark that this reduction is not optimal, as it is conjectured that $3$-$\OV$ requires $n^{3 - o(1)}$ time (also implied by $\SETH$). We include it here only as an illustration of the applicability of our reduction framework. It would be very interesting to improve it to a reduction from $3$-SUM to $\OV$\footnote{In~\cite{CarmosinoGIMPS16}, it is shown that under the $\textsf{NSETH}$ (which is controversial, see~\cite{Williams16_MA-SETH}), there is no fine-grained reduction from $\OV$ to $3$-SUM. But there is no formal evidence against the other direction.}.

Note that $3$-$\OV$ is actually a \emph{Satisfying-Triple} problem\footnote{It is also called a \emph{Product Space Problem} in~\cite{karthik2017parameterized}.}, and Theorem~\ref{theo:Sigma2-to-OV} only works for Satisfying-Pair problems. Still, we can generalize Theorem~\ref{theo:Sigma2-to-OV} easily to get the same connection between a 3-party $\Sigma_2$ communication protocol and a reduction from a satisfying-triple problem to $3$-$\OV$.

Let $F : \left(\{0,1\}^{d}\right)^{3} \to \{0,1\}$ be a function. $F\SATTRIPLE_{n}$ is the problem that you are given sets $A,B,C$ of $n$ vectors from $\{0,1\}^d$, and want to determine whether there is an $(a,b,c) \in A \times B \times C$ such that $F(a,b,c) = 1$. A 3-party $\Sigma_2$ communication protocol can be defined similarly as in Definition~\ref{defi:Sigma-2-communication-protocol} with the third player named Charles (we omit it here). We have the following analogous theorem of Theorem~\ref{theo:Sigma2-to-OV}.

\begin{theo}\label{theo:Sigma2-to-3OV}
	Let $F : \left(\{0,1\}^{d}\right)^{3} \to \{0,1\}$ and $n$ be an integer, suppose $F$ admits a computationally-efficient $\Sigma_2^\cc$ protocol. In which Merlin sends $m_1$ bits, Megan sends $m_2$ bits, Alice, Bob, and Charles communicate $\ell$ bits.
	
	Then there is a reduction from an $F\SATTRIPLE_{n}$ instance $I$ to $2^{m_1}$ $\thOV_{n,2^{(m_2+\ell)}}$ instances $J_1,J_2,\dotsc,J_{2^{m_1}}$, such that $I$ is a yes instance if and only if one of the reduced instances is a yes instance. And the reduction takes $O(n \cdot 2^{O(m_1 + m_2 + \ell)} \cdot \poly(d))$ time.
\end{theo}

We omit its proof here as it is identical to that of Theorem~\ref{theo:Sigma2-to-OV}.

Note that we can assume the integers in the $3$-SUM instance are in $[-n^4/2,n^4/2)$ without loss of generality. In order to apply Theorem~\ref{theo:Sigma2-to-3OV}, we need an efficient $\Sigma_2^\cc$ protocol for checking whether 3 numbers sum to zero.

\newcommand{\Fsum}{F_{\textsf{zero}}}

\begin{theo}\label{theo:protocol-for-checking-zero}
	For an integer $n$, let $\Fsum : \left( \{0,1\}^{4\log n} \right)^{3} \to \{0,1\}$ be the function that treats its inputs as three numbers in $[-n^4/2,n^4/2)$ (via a natural encoding), and checks whether they sum to zero. For any $1 \le T \le 4 \log n$, $\Fsum$ admits a $\Sigma_2^\cc$ protocol in which Merlin sends $O(\log n/T)$ bits, Megan sends $\lceil \log(n/T) \rceil$ bits and Alice, Bob, and Charles communicate $O(T)$ bits.
\end{theo}
\begin{proof}
	Let $x,y,z$ be three input numbers. Suppose Alice holds $x$, Bob holds $y$, and Charles holds $z$. We add $n^{4} / 2$ to each of them so that they now belong to $[0,n^4)$, and we want to check whether they sum up to $t = n^{4} / 2 \cdot 3$. Assuming $T$ divides $4 \log n$ for simplicity, we treat $x,y,z$ as numbers in $2^T$ base. Let $\ell = 4 \log n/ T$, and $x_{1},x_2,\dotsc,x_{\ell}$ be the digits of $x$ (from the least significant one to the most significant one, $y_i$'s, $z_i$'s, and $t_i$'s are defined similarly).
	
	Suppose we add $x,y,z$ together as numbers in $2^{T}$ base. Let $c \in \{0,1,2\}^{\ell}$ be a sequence of carries. We can see $x + y + z = t$ with respect to the carry sequence $c$ if and only if
	\[
	x_i + y_i + z_i + c_{i-1} = t_i + c_i \cdot 2^{T} \quad\text{for $i \in [\ell + 1]$.}
	\]
	In the above we set $c_0 = x_{\ell+1} = y_{\ell + 1} = z_{\ell + 1} = 0$.
	
	Therefore, in the protocol, Merlin sends the carry sequence $c$, which takes $O(\ell) = O(\log n / T)$ bits. Megan sends an index $i \in [\ell + 1]$. After that, Bob and Charles send $y_i$ and $z_i$ to Alice, respectively, and Alice accepts if and only if the above equality holds. It is straightforward to verify the protocol works.
\end{proof}

Finally, we are ready to prove Theorem~\ref{theo:3-sum-to-3-ov}.

\begin{proofof}{Theorem~\ref{theo:3-sum-to-3-ov}}
	Suppose $3$-$\OV$ is in truly-subquadratic time. That is, there is a constant $\epsilon$ such that for all constant $c$, $\thOV_{n,c\log n}$ can be solved in $n^{2-\epsilon}$ time.
	
	Given a $3$-SUM instance with integer entries in $[-n^4/2,n^4/2)$, it is just an $\Fsum\SATTRIPLE_{n}$ instance. Let $c_1$ be the constant hiding in $O(\log n/T)$ of Theorem~\ref{theo:Sigma2-to-3OV}, then $\Fsum$ admits a $\Sigma_2^\cc$ protocol in which Merlin sends $c_1\log n/T$ bits, Megan sends $\lceil \log(n/T) \rceil$ bits and Alice, Bob, and Charles communicate $O(T)$ bits.
	
	Set $T = c_1 \cdot 2/\epsilon$, Theorem~\ref{theo:Sigma2-to-3OV} implies this $3$-SUM instance can be reduced to $n^{\epsilon / 2}$ $\thOV_{n,2^{O(1/\epsilon)}\log n}$ instances. Applying the algorithm for $\thOV$, 3-SUM can be solved in $n^{2-\epsilon/2}$ time, which completes the proof.
\end{proofof}
